\documentclass[11pt,a4paper]{amsart}
\pdfoutput=1

\usepackage{amsmath}
	\numberwithin{equation}{section}
\usepackage{amssymb}
\usepackage{amsfonts}
 \usepackage{amsthm}
 	
 	\newtheorem{lemma}{Lemma}[section]

\usepackage{enumerate}

\usepackage{hyperref}

\usepackage{graphicx}
\usepackage{tikz}
	\usetikzlibrary{er,positioning}
	\usetikzlibrary{arrows.meta}
\newcommand{\R}{\mathbb{R}}
\newcommand{\C}{\mathbb{C}}
\newcommand{\Z}{\mathbb{Z}}

\newcommand{\ZU}{\mathcal{Z}_{\mathrm{u}}}
\newcommand{\ZUsymH}{\mathcal{Z}_{\mathrm{u},H} ^{\mathrm{sym.}}}
\newcommand{\ZUsymE}{\mathcal{Z}_{\mathrm{u},E} ^{\mathrm{sym.}}}
\newcommand{\ZguH}{\mathcal{Z}_{\mathrm{u},H} ^{\mathrm{gen.}}}
\newcommand{\ZguE}{\mathcal{Z}_{\mathrm{u},E} ^{\mathrm{gen.}}}

\newcommand{\mfsymH}{\mathcal{F}_{\mathrm{u},H}^{\mathrm{sym.}}}
\newcommand{\mfsymE}{\mathcal{F}_{\mathrm{u},E}^{\mathrm{sym.}}}
\newcommand{\mfgenuH}{\mathcal{F}_{\mathrm{u},H}^{\mathrm{gen.}}}
\newcommand{\mfgenuE}{\mathcal{F}_{\mathrm{u},E}^{\mathrm{gen.}}}

\newcommand{\ii}{\mathrm{i}}
\newcommand{\dd}{\mathrm{d}}

\begin{document}

\title[]{Exact equivalences and phase discrepancies between random matrix ensembles}
\author[]{Leonardo Santilli}
\address[LS]{Grupo de F\'{i}sica Matem\'{a}tica, Departamento de Matem\'{a}tica, Faculdade de Ci\^{e}ncias, Universidade de Lisboa, Campo Grande, Edif\'{i}cio C6, 1749-016 Lisboa, Portugal.}
\email{lsantilli@fc.ul.pt}

\author[]{Miguel Tierz}
\address[MT]{Departamento de Matem\'{a}tica, ISCTE - Instituto Universit\'{a}rio de Lisboa, Avenida das For\c{c}as Armadas, 1649-026 Lisboa, Portugal.}
\email{mtpaz@iscte-iul.pt}
\address[MT]{Grupo de F\'{i}sica Matem\'{a}tica, Departamento de Matem\'{a}tica, Faculdade de Ci\^{e}ncias, Universidade de Lisboa, Campo Grande, Edif\'{i}cio C6, 1749-016 Lisboa, Portugal.}
\email{tierz@fc.ul.pt}

\begin{abstract}
We study two types of random matrix ensembles that emerge when considering the same probability measure on partitions. One is the Meixner ensemble with a hard wall and the other are two families of unitary matrix models, with
weight functions that can be interpreted as characteristic polynomial insertions. We show that the models, while having the same exact evaluation for fixed values of the parameter, may present a different phase structure. We find phase transitions of the second and third order, depending on the model. Other relationships, via direct mapping, between the unitary matrix models and continuous random matrix ensembles on the real line, of Cauchy--Romanovski type, are presented and studied both exactly and asymptotically. The case of orthogonal and symplectic groups is studied as well and related to Wronskians of Chebyshev polynomials, that we evaluate at large $N$.
\end{abstract}

\maketitle
\tableofcontents

\section{Introduction}

Random matrix theory \cite{Mehtabook} has developed enormously, especially in the last two decades, attracting attention from researchers in a multitude of different fields \cite{Forrester,Blowerbook,AGZbook,AkeBaikDiFbook,PasturShbook,Taobook,BDSuidbook}. Indeed, one of its most exciting aspects is an inherent interdisciplinarity in the study of random matrices. Among the myriad of connections with other mathematical and physical areas, we have the important relationship with the study of random partitions, which in turn is deeply linked with problems in statistical mechanics, as we shall see.

In this work, we will see that random matrix ensembles which are related in different ways, either via direct mapping or emerging as two random matrix descriptions of the same object in the study of random partitions, have the same analytical evaluations for fixed values of the parameters yet, in some cases, the consideration of the large $N$ limit leads to a different phase structure. 

Originally, we started by noting that the analysis of probabilities given by a finite ensemble with weight of the Meixner type, for example as in \cite{CP:2013}, could, in principle, equally be studied with a random unitary matrix model, with a distribution of eigenvalues supported on the unit circle. Indeed, we could check, as we shall see, that the corresponding unitary matrix model gives the same results. Furthermore, the equivalence holds for the calculation of correlations near the edge of the eigenvalue density, in a double-scaling limit. 
However, when studying the more general ensemble with Meixner weight and a hard wall \cite{CP:2015} and the corresponding generalized unitary matrix model, while there is again an equivalence for fixed values of the parameter, a different asymptotic behaviour appears, with phase transitions of second order instead of third order. This discrepancy is a consequence of the complexification of the logarithm of the weight function of the more general unitary random matrix ensembles, studied in Section \ref{sec:Baikgen} below.

We further elaborate on this whole notion by studying other equivalences, with continuous random matrix ensembles on the real line, via direct mapping this time. We first introduce notions of random partitions and how either discrete random matrix ensembles, such as the Meixner ensemble with a hard wall, or continuous unitary random matrix ensembles appear. 

\medskip

	Let $\lambda = \left( \lambda_1, \lambda_2, \dots \right)$, $\lambda_1 \ge \lambda_2 \ge \dots$, be a partition, and $\mathfrak{s}_{\lambda} (\underline{t})$ the Schur polynomial associated to it \cite{Macdonaldbook}, evaluated at $\underline{t} =(t_1, t_2, \dots)$. The Schur measure is the probability measure on the space of partitions defined as \cite{Owedge}
	\begin{equation}
	\label{eq:Schurmeasure}
		\mathrm{Prob}_{(\underline{t} , \underline{t} ^{\prime} )} \left( \lambda \right) = \frac{1}{\mathcal{Z}_{\mathrm{Schur}} (\underline{t} , \underline{t} ^{\prime} ) }  \mathfrak{s}_{\lambda} ( \underline{t} ) \mathfrak{s}_{\lambda} ( \underline{t} ^{\prime} ) ,
	\end{equation}
	with normalization constant the inverse of the partition function
	\begin{equation*}
		\mathcal{Z}_{\mathrm{Schur}} (\underline{t} , \underline{t} ^{\prime} ) = \sum_{\lambda}  \mathfrak{s}_{\lambda} ( \underline{t} ) \mathfrak{s}_{\lambda} ( \underline{t} ^{\prime} ) = \prod_{j,k} \left( 1 - t_i t^{\prime} _j \right)^{-1} ,
	\end{equation*}
	where last equality is the Cauchy identity. Here we are denoting $\mathrm{Prob}_{(\underline{t} , \underline{t} ^{\prime} )}$ the probability taken with respect to a given choice of parameters $(\underline{t} , \underline{t} ^{\prime} )$. Define the set
	\begin{equation*}
		\mathfrak{S} ( \lambda) := \left\{ \lambda_j - j + \frac{1}{2} , \ j=1,2, \dots \right\} \subset \mathbb{Z} + \frac{1}{2} ,
	\end{equation*}
	which encodes the shape of the partition $\lambda$. For a given subset $\mathfrak{X}  \subset \mathbb{Z} + \frac{1}{2}$, the probability
	\begin{equation*}
		\mathrm{Prob}_{(\underline{t} , \underline{t} ^{\prime} )} \left( \mathfrak{X} \right) = \frac{1}{\mathcal{Z}_{\mathrm{Schur}} (\underline{t} , \underline{t} ^{\prime} ) }  \sum_{\lambda  \vert  \mathfrak{S} (\lambda) \supset \mathfrak{X} } \mathfrak{s}_{\lambda} ( \underline{t} ) \mathfrak{s}_{\lambda} ( \underline{t} ^{\prime} ) 
	\end{equation*}
	that, picking a random partition $\lambda$ according to the Schur measure we get $\mathfrak{X} \subset \mathfrak{S} (\lambda) $, has a determinantal representation. Such representation is the determinant of a known correlation kernel admitting an integral representation \cite{Owedge} (see also \cite[Sec. 4]{OLect}).\par
	In the present work, we are interested in a special case of Schur measure, which is a particular sub-case of what is known in the literature as $z$-measure \cite{BorodinOl,OSL2}. For any given $n \in \mathbb{N}$ the $z$-measure is a probability measure on the partitions $\left\{ \lambda \right\} $ of $n$ (thus with $\vert \lambda \vert =n$) depending on two parameters $z$ and $z^{\prime}$. A probability measure over all partitions $\left\{ \lambda \right\}$ (with $\vert \lambda \vert$ arbitrary) is obtained passing to the grand canonical ensemble, summing over $n$. In defining the grand canonical ensemble, every $n \in \mathbb{N}$ is weighted with a negative binomial distribution \cite{BorodinOl}.\par
	We choose the parameters in \eqref{eq:Schurmeasure} to be
	\begin{equation}
	\label{eq:zmeasureparam}
		\underline{t} =  ( \underbrace{ t, \cdots , t }_{N_{1} \text{ times}} , 0 ,\dots ) , \quad \underline{t} ^{\prime} = ( \underbrace{ t, \cdots , t }_{N_{2} \text{ times}} , 0 ,\dots ) , \qquad 0<t<1 ,
	\end{equation}
	so that the Schur measure \eqref{eq:Schurmeasure} becomes a particular instance of $z$-measure with $z=N_1, z^{\prime} =N_2 \in \mathbb{N}$, and the parameter of the binomial distribution being $t^2$. This choice of $z,z^{\prime}$ is the degenerate case, \emph{i.e.}, non-negative measure, and positive definite when restricted to partitions of length at most $N_1$. The probabilities
	\begin{equation*}
		\mathrm{Prob}_{t} \left( \mathfrak{X} \right) = \left( 1 - t^2 \right)^{N_1 N_2} \sum_{\lambda  \vert  \mathfrak{S} (\lambda) \supset \mathfrak{X} } \mathfrak{s}_{\lambda} ( \underline{t} ) \mathfrak{s}_{\lambda} ( \underline{t} ^{\prime} )  ,
	\end{equation*}
	where by $\mathrm{Prob}_{t}$ we understand the probability taken with specialization of parameters \eqref{eq:zmeasureparam}, admit a determinantal representation in terms of the so-called \emph{hypergeometric kernel} \cite{BorodinOl}. The properties of the hypergeometric function were used in \cite{BorodinOl} to show that, when $z=N_1, z^{\prime}=N_2$ are integers, the hypergeometric kernel becomes proportional to the Meixner kernel. For a collection of results about the $z$-measure, see the survey \cite{BOreview}, and in particular Proposition 6.1 therein for the connection with the Meixner ensemble.\par
	In the present work we are interested in the quantities
	\begin{equation}
	\label{eq:ZHschur}
		\mathcal{Z}_{H} (t) = \sum_{\lambda }   \mathfrak{s}_{\lambda} ( \underline{t} ) \mathfrak{s}_{\lambda} ( \underline{t} ^{\prime} ) = \sum_{\lambda} \left( \dim \lambda \right)^2 t^{2 \left\vert \lambda \right\vert } =\left( 1 - t^2 \right)^{-N_{1} N_{2}} 
	\end{equation}
	and (the sum is over all $\lambda$ with $\lambda_1 \le K$)
	\begin{equation}
	\label{eq:ZEschur}
		\mathcal{Z}_{E} (t) = \sum_{\lambda \vert \lambda_1 \le K }   \mathfrak{s}_{\lambda} ( \underline{t} ) \mathfrak{s}_{\lambda} ( \underline{t} ^{\prime} ) = \sum_{\lambda \vert \lambda_1 \le K }  \left( \dim \lambda \right)^2 t^{2 \left\vert \lambda \right\vert } ,
	\end{equation}
	with choice of parameter understood to be as in \eqref{eq:zmeasureparam}, and $\vert \lambda \vert = \sum_j \lambda_j$ is the size of the partition $\lambda$, corresponding to the total number of boxes in its diagram, and $\dim \lambda$ is the dimension of the irreducible representation of the symmetric group $S_{\vert \lambda \vert}$ labelled by $\lambda$. The meaning of the subscripts $H$ and $E$ will be clear in a moment. Notice that, due to the property $\mathfrak{s}_{\lambda} ( \underline{t}) = 0$ if $\mathrm{length} ( \lambda ) > \mathrm{length} ( \underline{t}) $, the sums are effectively truncated to partitions of length at most $N_1$. Furthermore, the ratio
	\begin{equation*}
		\frac{ \mathcal{Z}_{E} (t) }{ \mathcal{Z}_{H} (t) } = \mathrm{Prob}_t \left( \lambda_1 \le K \right)
	\end{equation*}
	is the probability that, picking a random partition $\lambda$ with probability distribution as described above, the length of its rows is at most $K$.\par
	The $z$-measure induces a determinantal point process on $\mathbb{Z} + \frac{1}{2}$, thus the correlation functions have determinantal form
	\begin{equation*}
		\mathrm{Prob}_{t} \left( \mathfrak{X} \subset \mathfrak{S} ( \lambda) \right) = \det \left( \mathcal{K} \right)_{\mathfrak{X}} , 
	\end{equation*}
	for $\mathfrak{X} \subset \mathbb{Z} + \frac{1}{2}$, where $\mathcal{K}$ is the operator whose kernel, the function  $\mathcal{K} (x,y)$ on $\left( \mathbb{Z} + \frac{1}{2} \right) \times \left( \mathbb{Z} + \frac{1}{2} \right)$, is the hypergeometric kernel. Therefore we have a Fredholm determinant representation of $\mathcal{Z}_{E} (t)$:
	\begin{align*}
		\mathcal{Z}_{E} (t) & = (1-t^2)^{-N_1 N_2} \mathrm{Prob}_t \left( \lambda_1 \le K \right) \\
		& = (1-t^2)^{-N_1 N_2} \mathrm{Prob}_t \left( \left\{ x \in \mathbb{Z} + \frac{1}{2} \ \vert \ x \le K - \frac{1}{2} \right\} \cap \mathfrak{S} ( \lambda ) = \emptyset \right) \\
		& = ( 1-t^2 )^{-N_1 N_2} \det \left(  \mathbf{1} - \mathcal{K} \right)_{ \mathbb{Z}_{\ge K} + \frac{1}{2}  } ,
	\end{align*}
	where $\mathbf{1}$ in the last line is the identity operator. Using the hook-length formula 
	\begin{equation*}
		\dim \lambda = \prod_{j < k} \frac{ \lambda_j - \lambda_k - j + k }{ k-j} 
	\end{equation*}
	 for the dimension in eq. \eqref{eq:ZEschur} and changing variables $h_j  = \lambda_j - j + N$, we arrive to the expression 
	\begin{align}
	\label{eq:ZEMeix}
		\mathcal{Z}_{E} (t) & = \frac{ t^{-N_1 (N_1-1)} \left( \Gamma (N_2 -N_1 +1) \right)^{N_1} }{ N_1 ! G(N_1+1) G(N_2 +1) } \\
		& \times \sum_{ h_1 = 0 } ^{N_1 + K -1} \cdots  \sum_{ h_{N_{1}} = 0 } ^{N_1 + K -1} \prod_{1 \le j < k \le N_1} \left( h_j - h_k \right)^2 \prod_{j=1} ^{N_1} \left( \begin{matrix} N_2 - N_1 + h_j \\ h_j  \end{matrix} \right) t^{2 h_j} , \notag
	\end{align}
	where we also used the symmetry in the $h_j$ variables to remove the restriction to the Weyl chamber $\lambda_1 \ge \lambda_2 \ge \dots \ge \lambda_{N_1}$. $G( \cdot )$ is the Barnes $G$-function \cite{barnes1899theory}, which, for integer values of the argument is $G(n) = \prod_{j=0} ^{n-2} j!$. 
	The expression \eqref{eq:ZEMeix} is a Meixner ensemble with summation restricted to $0 \le h_j \le N_1 + K-1$ and coincides with the partition function of the dimer model studied in \cite{CP:2013,CP:2015}\footnote{The dimer model is associated with the random tiling of an Aztec diamond with a square \cite{CP:2013} or a rectangle \cite{CP:2015} cut off. See \emph{e.g.} the recent work \cite{stephan2020} for an overview on dimers, random tilings and random matrices.}, and has also appeared in \cite{Johansson:2010}. In the Coulomb gas picture, the restriction in the summation range corresponds to a hard wall for the charges (\emph{i.e.} an infinite barrier) placed at $N_1+K-1$. We assumed $N_1 \le N_2$, but a completely analogous expression can be easily obtained in the converse case.\par
	Let us introduce the generating functions of, respectively, the complete homogeneous polynomials $\left\{ \mathfrak{h}_k \right\} $ and the elementary symmetric polynomials $\left\{ \mathfrak{e}_k \right\}$, specialized at $\underline{t}$:
	\begin{align*}
		H (z; \underline{t}) &= \sum_{k=0} ^{\infty} \mathfrak{h}_k ( \underline{t} ) z^k =  \prod_{k} \left( 1-t_k z \right)^{-1} , \\
		E (z; \underline{t}) &= \sum_{k=0} ^{\infty} \mathfrak{e}_k ( \underline{t} ) z^k =  \prod_{k} \left( 1 + t_k z \right) .		
	\end{align*}
	The partition functions \eqref{eq:ZHschur}-\eqref{eq:ZEschur} admit a determinantal representation in terms of determinants of $K \times K$ Toeplitz matrices, with symbol, respectively
	\begin{align*}
		\sigma_{H} (z,t) & = H(z; \underline{t}) H(z^{-1}; \underline{t}^{\prime}) = \left( 1 - t z \right)^{-N_1}   \left( 1 - t z^{-1} \right)^{-N_2}  , \\
		\sigma_{E} (z,t) & = E(z; \underline{t}) E(z^{-1}; \underline{t}^{\prime}) = \left( 1 + t z \right)^{N_1}   \left( 1 + t z^{-1} \right)^{N_2}  .
	\end{align*}
	 See for example \cite{Hikami:2003,GGT1} for the explicit derivation (see also \cite{BDSuidbook} for an extensive account of Toeplitz determinants and their properties). In turn, using Andr\'{e}ief's identity \cite{Andreief,MeetA} we have that these Toeplitz determinants admit a representation as unitary matrix integrals 
	\begin{align}
		\mathcal{Z}_{H} (t) &= \frac{1}{K!} \int_{[-\pi,\pi]^K} \frac{ \dd ^K \varphi }{(2 \pi)^K} \prod_{1 \le j < k \le K} \lvert e^{\ii \varphi_j} - e^{\ii \varphi_k} \rvert \prod_{j=1} ^K \left( 1-t e^{\ii \varphi_j} \right)^{- N_1} \left( 1-t e^{- \ii \varphi_j} \right)^{- N_2 } \label{eq:ZHmm},	\\
		\mathcal{Z}_{E} (t) &= \frac{1}{K!} \int_{[-\pi,\pi]^K} \frac{ \dd ^K \varphi }{(2 \pi)^K} \prod_{1 \le j < k \le K} \lvert e^{\ii \varphi_j} - e^{\ii \varphi_k} \rvert \prod_{j=1} ^K \left( 1+t e^{\ii \varphi_j} \right)^{N_1} \left( 1+t e^{- \ii \varphi_j} \right)^{N_2 } \label{eq:ZEmm} .
	\end{align}
	We will refer to \eqref{eq:ZHmm} and \eqref{eq:ZEmm} as the $H$-model and the $E$-model, respectively. Therefore we have two equivalent matrix model descriptions of the quantity $\mathcal{Z}_{E} (t)$, or, equivalently, of the probability $\mathrm{Prob}_t ( \lambda_1 \le K )$: either as a discrete matrix model on the bounded subset $\left\{ 0, 1 , \dots, N_1 + K -1 \right\} \subset \mathbb{Z}$, or as a continuous matrix model on the unit circle.\par
	We stress that the equivalence between these two representations does not rely on a direct map, but rather on a two-step procedure relating the two matrix model formulations of $\mathcal{Z}_E (t)$ to the same Toeplitz determinant. As a consequence, the quantity $( N_1 - N_2 )$, which has the meaning of a deformation parameter, plays different roles in the two pictures. This will be reflected in the mismatch of the phase structure when $N_1, N_2 \to \infty$ with $ N_1 - N_2  \ne 0$.

\section{Random matrix ensembles on the unit circle}
\label{sec:BaikMM}

	In the present Section, we consider the asymptotic behaviour of the unitary matrix models defined in eq. \eqref{eq:ZHmm} and \eqref{eq:ZEmm}, when the rank $K$ is large and $N_1, N_2$ scale with $K$. Before that, we comment on the already known aspects of the exact solvability of some of the models above. 
	
	\subsection{Prologue: Exact evaluation}
	\label{sec:unitexact}
		Following \cite{Johansson:2000}, the authors of \cite{CP:2015} gave an explicit evaluation of the discrete matrix model \eqref{eq:ZEMeix} at the limit value $t=1$. The exact formula of \cite[Proposition 3.1]{CP:2015} was obtained thanks to the fact that at the limit value $t=1$ the Meixner ensemble with a hard wall becomes a Hahn ensemble.
		
		On the unitary matrix model side, at $t=1$ (or more generally $\lvert t \rvert =1$) the symbol $\sigma_E$ develops a Fisher-Hartwig singularity, and it can be evaluated exactly thanks to a formula by B\"{o}ttcher and Silbermann \cite{BS:1985} (see also \cite{BSshort} for another proof):
		\begin{equation}
		\label{eq:exactZEt1}
			\mathcal{Z}_{E} (t) = \frac{ G (N_1 + 1 ) G(N_2 + 1 ) G (N_1 + N_2 + K +1) }{ G (N_1 + N_2 + 1) G (N_1 + K + 1) G( N_2 + K +1) } G(K+1) .
		\end{equation}
		This provides an exact check of the equivalence. We will provide a third independent derivation of this result later in Section \ref{sec:exactevZEreal}. For ease of the reader, we report in table \ref{tab:dictioCP} the correspondence between the notation of \cite{CP:2013,CP:2015}, that in \cite{Johansson:2000} and ours.
		\begin{table}
		\centering
		\begin{tabular}[bht]{l | c | c | c | c | c | c}
			In \cite{CP:2015} & $s$ & $r$ & $q$ & $\alpha $ & $R$ & $Q$ \\
			\hline
			In \cite{Johansson:2000} & $N$ & $N +t$ & $M -N = [\gamma N ] - N$ & $q $ & $\omega +1$ & $\gamma -1$ \\
			\hline
			Here & $N_1$ & $N_1 + K$ & $N_2 - N_1$ & $t^2$ & $1 + \gamma^{-1}$ & $2v /(1-v)$
		\end{tabular}
		\caption{Dictionary between the notation in \cite{CP:2015}, in \cite{Johansson:2000} and the present work.}
		\label{tab:dictioCP}
		\end{table}\par
		Note that, if we choose the symmetric model $N_1=N_2 \equiv N$ and modify the symbol by inserting a monomial factor $z^{s}$, with $s \in \mathbb{Z}$, the exact evaluation of the corresponding matrix model at $t=1$ is again given by the formula of \cite{BS:1985,BSshort}, thanks to the simple identity 
		\begin{equation}
		\label{unitaryshiftsymb}
			z^s \sigma_E (z, t=1) = z^{s} \left( 1+z \right)^{\beta} \left( 1+z^{-1} \right)^{\beta} = \left( 1+z \right)^{\beta +s}  \left( 1+z^{-1} \right)^{\beta -s} , 
		\end{equation}
		when $z=e^{\ii \varphi}$, $ - \pi \le \varphi \le \pi$. The effect of the monomial insertion is to shift the Fourier coefficients of the symbol $\sigma_E (z; 1)$ by $s$: the $k^{\mathrm{th}}$ Fourier coefficient of the symbol $z^s \sigma_E (z; 1)$ is the $(k+s)^{\mathrm{th}}$ coefficient of $\sigma_E (z; 1)$. Therefore, we can evaluate exactly the partition function of the unitary ensemble with weight \eqref{unitaryshiftsymb}, and it is again given by formula \eqref{eq:exactZEt1}, with $N_1=N+s, N_2 =N -s$. \par
		It is worth mentioning that there are a wealth of analytical results on Toeplitz banded matrices, whose determinant is given by the $E$-models above \cite{bottcher2005spectral}. More generally, there are many analytical results on determinants of Toeplitz matrices generated by a rational symbol, but we will explore complementary approaches here and indicate possible open problems at the end, in the Outlook Section.\par
		We just quote here that for \eqref{eq:ZEmm} with $N_1=N_2=1$ we have the determinant of a symmetric tridiagonal Toeplitz matrix, known to be equal to a Chebyshev polynomial of the second type: 
		\begin{equation*}
		    \mathcal{Z}_{E} (t) \vert_{N_1 = N_2 = 1} = t^{K} U_{K}\left( \frac{1+t^{2}}{2t}\right)= \sum_{n=0} ^{K} t^{2n} .
    	\end{equation*}
    	
	\subsection{Unitary matrix models}
	\label{sec:unitmat}
	
	Comparing the definition \eqref{eq:ZEschur} of the $E$-model with \eqref{eq:ZHschur}, one sees that its $K \to \infty $ limit, with fixed $N_1, N_2$ coincides with the $H$-model:
	\begin{equation*}
		\lim_{K \to \infty} \mathcal{Z}_{E} (t) = \mathcal{Z}_H (t) .
	\end{equation*}
	In the scaling limit with $N_1, N_2$ growing together with $K$, it was found in \cite{CP:2013,CP:2015} that the presence of hard walls in the discrete matrix model \eqref{eq:ZEMeix} triggers a third order phase transition. For the case $N_1=N_2$, starting from a result of Baik \cite{Baik:2000} we will prove the phase transition from the point of view of the unitary matrix model. This is done in Section \ref{sec:Baiksym}. For the general case $N_1 \ne N_2$, however, the potential of the unitary matrix model is complex-valued, and the large $K$ asymptotic becomes more involved. This topic is analyzed in Section \ref{sec:Baikgen}.\par
	As the asymptotic behaviour at large $K$ does not depend on $N_1, N_2$ being integers, we consider the matrix models arising from Toeplitz determinants with more general symbols
	\begin{align*}
		\sigma_H (z;t) & = \left( 1-t z \right)^{- \beta_1} \left( 1-t z^{-1} \right)^{- \beta_2} , \\
		\sigma_E (z;t) & = \left( 1+t z \right)^{\beta_1} \left( 1+t z^{-1} \right)^{\beta_2} ,
	\end{align*}
	and, without loss of generality, we assume $0 \le \beta_1 \le \beta_2$. We also introduce the notation
	\begin{equation}
	\label{eq:defbetaandv}
		\beta_1 = \beta (1-v) , \quad  \beta_2 = \beta (1+v) ,
	\end{equation}
	where $\beta =(\beta_1 + \beta_2)/2$ is the average power and $0 \le v \le 1$ measures the asymmetry in $z \leftrightarrow z^{-1}$. We then take the $K \to \infty$ limit with\footnote{\label{foot:gamma}This is the inverse of the usual 't Hooft coupling as defined in gauge theories, but here we adopt to the notation of \cite{BDJ:1999,Baik:2000}. Consistently, ``weak'' and ``strong'' coupling will refer to the values of $\gamma$.}
	\begin{equation}
	\label{eq:defgamma}
		\gamma := \frac{\beta}{K} \ \text{ fixed} .
	\end{equation}
	As customary, when studying the large rank behaviour of matrix models, we introduce the density of eigenvalues
	\begin{equation}
	\label{eq:defrho}
		\rho (\varphi ) = \frac{1}{K} \sum_{j=1} ^K \delta ( \varphi - \varphi_j ),
	\end{equation}
	which at large $K$ becomes a continuous function of $\varphi$, with compact support and normalized so that
	\begin{equation*}
		\int_{-\pi} ^{\pi} \dd \varphi \rho ( \varphi ) = 1 .
	\end{equation*}
	In each of the cases considered below, we will find the eigenvalue density $\rho$ and use it to evaluate the free energy in this limit, defined as:
	\begin{equation*}
		\mathcal{F} := - \frac{1}{K^2} \log \mathcal{Z} .
	\end{equation*}\par
	Our calculations are based on standard saddle point techniques, and we omit them from the main text and refer to the Appendices \ref{app:saddlecalc} and \ref{app:FEcalc}. We will show that all the models undergo a phase transition when a gap opens in the support of the eigenvalue density, as schematized in Figure \ref{fig:gapPT}.
	
		\begin{figure}[bth]
			\centering	\includegraphics[width=0.28\textwidth]{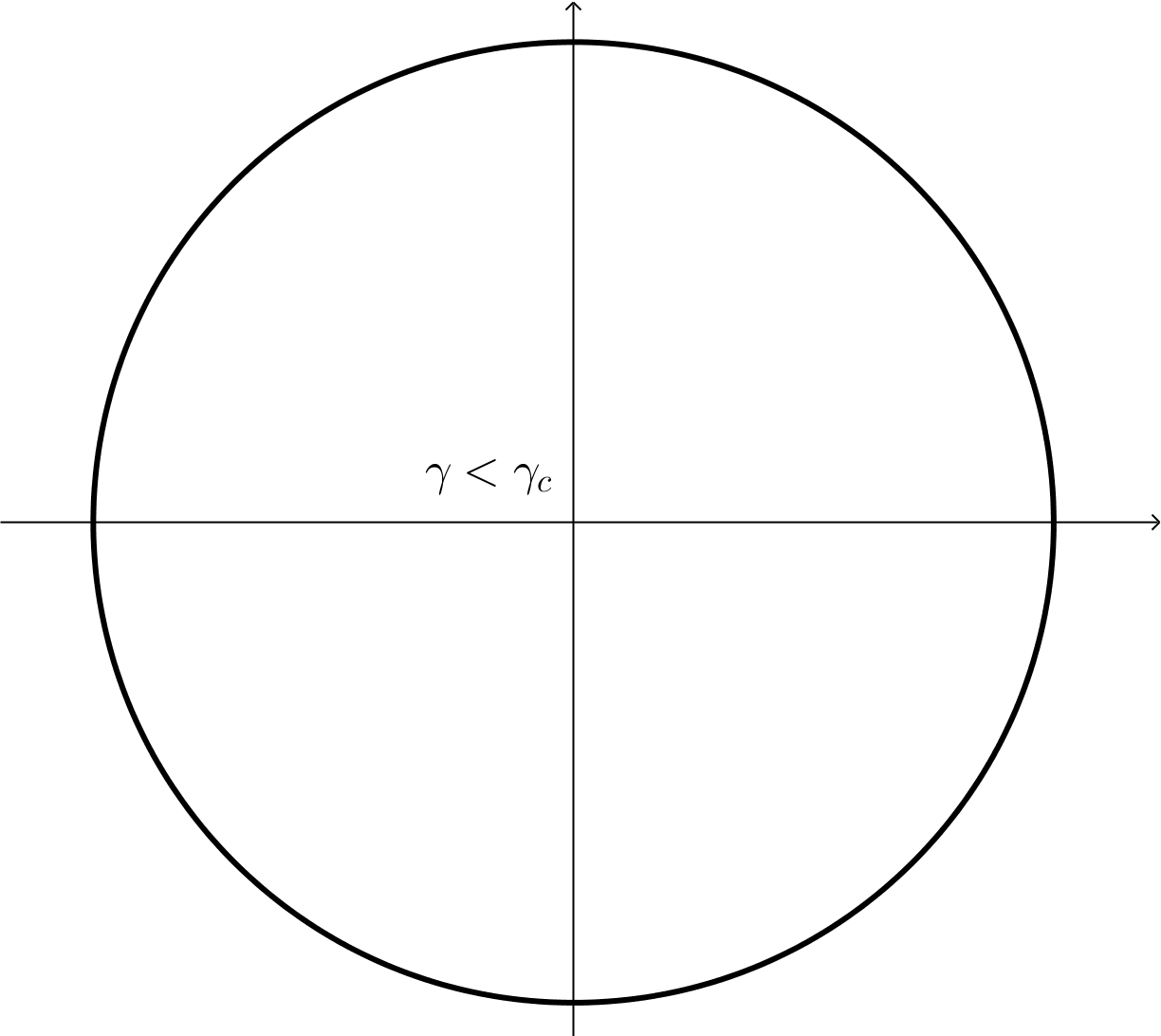}%
	\hspace{0.1\textwidth}	\includegraphics[width=0.28\textwidth]{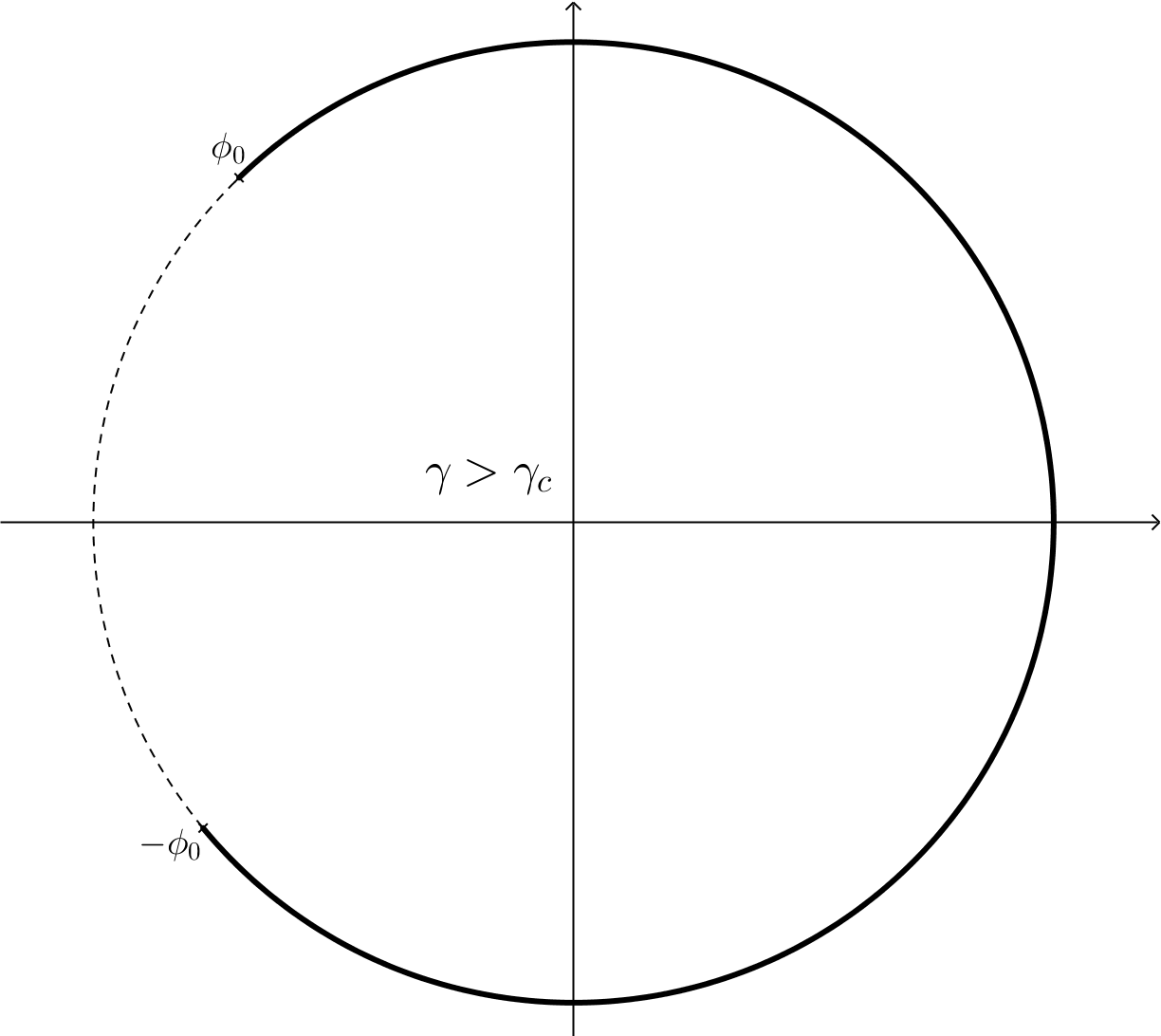}
	\caption{Increasing the coupling $\gamma$, the support of the eigenvalue density $\rho$ develops a gap, signalling a phase transition.}
		\label{fig:gapPT}
		\end{figure}
	
	\subsection{Phase transition: symmetric case}
	\label{sec:Baiksym}
	
		We first focus on the symmetric case $\beta_1=\beta_2 \equiv \beta$, while the analysis of the more general case $\beta_1 \ne \beta_2$ is undertaken later in Section \ref{sec:Baikgen}.\par
		Hence, the matrix models we analyze are:
		\begin{align}
				\ZUsymH & =  \frac{1}{K!} \int_{[- \pi, \pi]^K} \frac{ \dd^K \varphi }{ (2 \pi )^K } \prod_{1 \le j < k \le K} \left\lvert e^{\ii \varphi_j } - e^{\ii \varphi_k} \right\rvert^2 \prod_{j=1} ^K \left[ (1-t e^{\ii \varphi_j})(1-t e^{- \ii \varphi_j}) \right]^{-\beta} , \label{eq:BaikH} \\
				\ZUsymE & =  \frac{1}{K!} \int_{[- \pi, \pi]^K} \frac{ \dd^K \varphi }{ (2 \pi )^K } \prod_{1 \le j < k \le K} \left\lvert e^{\ii \varphi_j } - e^{\ii \varphi_k} \right\rvert^2 \prod_{j=1} ^K \left[ (1+t e^{\ii \varphi_j})(1+t e^{- \ii \varphi_j}) \right]^{\beta} ,	\label{eq:BaikE}	
		\end{align}
		and we recall that $\ZUsymH$ admits an exact solution through the Cauchy identity, whilst $\ZUsymE$ does not. Nevertheless we have $\lim_{K \to \infty} \ZUsymE = \ZUsymH$.\par
		In \cite{Baik:2000}, Baik proved that the first system, described by the partition function \eqref{eq:BaikH}, and which we will call for simplicity the $H$-model, undergoes a phase transition at large $K$. We prove that the second system \eqref{eq:BaikE}, which we call $E$-model, undergoes the same phase transition. We prove it solving a singular integral equation in Appendix \ref{app:saddlecalc}, but in fact the result may also be directly obtained from \cite{Baik:2000}, with minor changes.
		
		\subsubsection{The $H$-model}
			Consider the matrix integral in \eqref{eq:BaikH}, and take the limit $K \to \infty $ with $\beta$ scaling as in \eqref{eq:defgamma}. The leading contribution to $\ZUsymH$ comes from the solution to the system of saddle point equations that, with the help of the eigenvalue density $\rho_H$ as defined in \eqref{eq:defrho}, can be rewritten as a single singular integral equation:
			\begin{equation}
			\label{SPEBaikH}
				- \ii \gamma t \left[ \frac{ e^{\ii \varphi} }{1-t e^{\ii \varphi} } - \frac{ e^{- \ii \varphi} }{1-t e^{- \ii \varphi} } \right] = \mathrm{P} \int \dd \vartheta \rho_H (\vartheta) \cot \left( \frac{ \varphi - \vartheta }{2} \right) ,
			\end{equation}
			where the symbol $\mathrm{P} \int$ means principal value of the integral, and $\rho_H$ is the eigenvalue density for the specific model considered presently. The details of the solution to eq. \eqref{SPEBaikH} are spelled in Appendices \ref{app:ZUHweak} and \ref{app:ZUHstrong}. Two phases exists, separated by the critical curve \cite{Baik:2000}
			\begin{equation*}
				\gamma = \frac{1+t}{2t} =: \gamma_{c,H} (t) .
			\end{equation*}
			The eigenvalue density, plotted in Figure \ref{fig:plotrhoH} for various $t$ and $\gamma$, reads:
			\begin{equation}
				\rho_H (\varphi) = \begin{cases} \frac{1}{2\pi} \left[ 1 + 2 \gamma t \left( \frac{\cos \varphi - t }{ (1 -t)^2 +4t \left(\sin \frac{ \varphi}{2} \right)^2 } \right)  \right]  , &  \gamma \le \gamma_{c,H} (t) , \\ \frac{2 (\gamma-1) t }{\pi} \left( \frac{ \cos \frac{\varphi}{2} }{ (1-t)^2 +4t \left( \sin \frac{\varphi}{2} \right)^2 } \right)\sqrt{ \left( \sin \frac{\phi_0}{2} \right)^2 -  \left( \sin \frac{\varphi}{2} \right)^2  } , & \gamma > \gamma_{c,H} (t) \end{cases} 
			\end{equation}\par
		and allows to evaluate the free energy $\mfsymH$, obtaining:
			\begin{equation}
			\label{eq:FEBaikH}
			  \mfsymH = \begin{cases}  - \gamma^2 \log (1-t^2) ,  & \gamma \le \gamma_{c,H} (t) , \\ - (2 \gamma-1) \log (1-t) - \frac{1}{2} \log t + \mathcal{C}_{H} (\gamma) ,  & \gamma > \gamma_{c,H} (t) , \end{cases} 
			\end{equation}
			where $\mathcal{C}_{H} (\gamma)$ is $t$-independent. See Appendices \ref{app:Fbothweak} and \ref{app:Fbothstrong} for the calculation of the free energy. From the latter expression one sees that $\frac{ \dd \mfsymH }{ \dd t} $ and $\frac{ \dd^2 \mfsymH }{ \dd t^2} $ are continuous functions for all values of $\gamma$, while
			\begin{equation*}
				\lim_{\gamma \uparrow \gamma_{c,H} } \frac{ \dd^3 \mfsymH }{ \dd t^3} - \lim_{\gamma \downarrow \gamma_{c,H} } \frac{ \dd^3 \mfsymH }{ \dd t^3} = \frac{1}{t^3 (1-t^2)} ,
			\end{equation*}
			and therefore the system undergoes a third order phase transition at the critical curve $\gamma_{c,H} (t)= \frac{1+t}{2t}$.

		\begin{figure}[bth]
		\centering
		\includegraphics[width=0.3\textwidth]{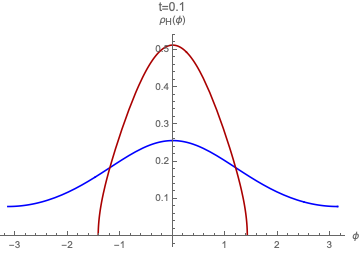}%
		\includegraphics[width=0.3\textwidth]{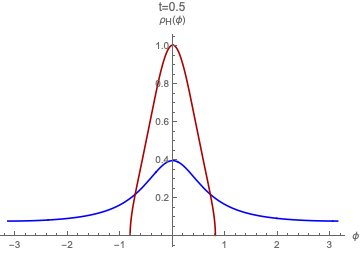}%
		\includegraphics[width=0.3\textwidth]{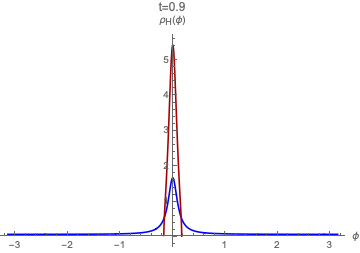}
		\caption{Eigenvalue density $\rho_H (\phi)$. The blue curve is at $\gamma = \frac{1}{2} \gamma_{c,H} (t)$ and the red curve is at $\gamma = 2 \gamma_{c,H} (t)$, for $t=$ 0.1 (left), 0.5 (center), 0.9 (right).}
		\label{fig:plotrhoH}
		\end{figure}
			
		\subsubsection{The $E$-model}
			We now turn to the second matrix model, defined in eq. \eqref{eq:BaikE}. The leading contribution in the large $K$ limit, with scaling \eqref{eq:defgamma}, is obtained solving the saddle point equation
			\begin{equation}
			\label{SPEBaikE}
				- \ii \gamma t \left[ \frac{ e^{\ii \varphi} }{1+t e^{\ii \varphi} } - \frac{ e^{- \ii \varphi} }{1+t e^{- \ii \varphi} } \right] = \mathrm{P} \int \dd \vartheta \rho_E (\vartheta) \cot \left( \frac{ \varphi - \vartheta }{2} \right) .
			\end{equation}
			We solve this singular integral equation in Appendices \ref{app:ZUEweak} and \ref{app:ZUEstrong}. From direct comparison of the integral representation of $\ZUsymH$ and $\ZUsymE$ in \eqref{eq:BaikH} and \eqref{eq:BaikE}, one would expect that, the solution to the second model is related to the solution to the first model by
			\begin{equation*}
				(\gamma, t) \mapsto (- \gamma, -t) .
			\end{equation*}
			Direct calculations prove that this is true and, in particular, the system undergoes a phase transition along the critical curve
			\begin{equation*}
				\gamma = \frac{1-t}{2t} =: \gamma_{c,E} (t) .
			\end{equation*}
			The eigenvalue density in the $E$-model, plotted in Figure \ref{fig:plotrhoE} for different values of $t$ and $\gamma$, is:
			\begin{equation*}
				\rho_E (\varphi) = \begin{cases} \frac{1}{2\pi} \left[ 1 + 2 \gamma t \left( \frac{\cos \varphi + t }{ (1 +t)^2 -4t \left(\sin \frac{ \varphi}{2} \right)^2 } \right) \right] , & \gamma \le \gamma_{c,E} (t) , \\ \frac{2 (\gamma+1) t }{\pi} \left( \frac{ \cos \frac{\varphi}{2} }{ (1+t)^2 -4t \left( \sin \frac{\varphi}{2} \right)^2 } \right)\sqrt{ \left( \sin \frac{\phi_0}{2} \right)^2 -  \left( \sin \frac{\varphi}{2} \right)^2  } , & \gamma > \gamma_{c,E} (t) , \end{cases} 
			\end{equation*}
			which allows to compute the free energy $\mfsymE $ at large $K$ (see Appendices \ref{app:Fbothweak} and \ref{app:Fbothstrong}), and the final result is:
			\begin{equation}
			\label{eq:FEBaikE}
			  \mfsymE = \begin{cases}  - \gamma^2 \log (1-t^2) , &  \gamma \le \gamma_{c,E} (t) , \\  (2 \gamma+1) \log (1+t) - \frac{1}{2} \log t + \mathcal{C}_{E} (\gamma) , & \gamma > \gamma_{c,E} (t) , \end{cases} 
			\end{equation}
			where $\mathcal{C}_{E} (\gamma)$ is $t$-independent. Taking derivatives, one finds again that first and second derivatives are continuous functions, while
			\begin{equation*}
				\lim_{\gamma \uparrow \gamma_{c,E} } \frac{ \dd^3 \mfsymE }{ \dd t^3} - \lim_{\gamma \downarrow \gamma_{c,E} } \frac{ \dd^3 \mfsymE }{ \dd t^3} = \frac{1}{t^3 (1-t^2)} ,
			\end{equation*}
			thus the phase transition is of third order also in this case.
			
				\begin{figure}[bth]
		\centering
		\includegraphics[width=0.3\textwidth]{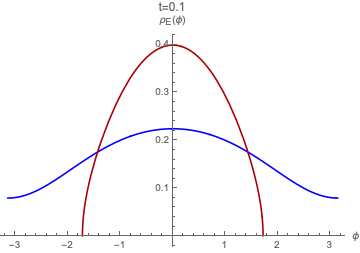}%
		\includegraphics[width=0.3\textwidth]{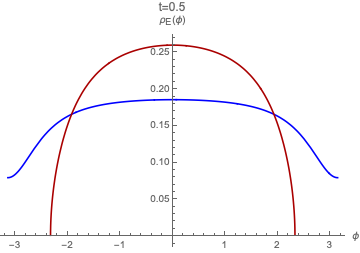}%
		\includegraphics[width=0.3\textwidth]{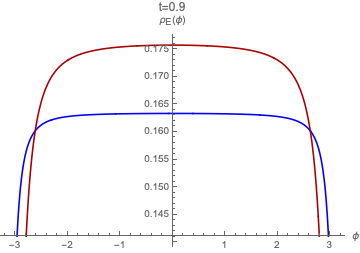}
		\caption{Eigenvalue density $\rho_E (\phi)$. The blue curve is at $\gamma = \frac{1}{2} \gamma_{c,E} (t)$ and the red curve is at $\gamma = 2 \gamma_{c,E} (t)$, for $t=$ 0.1 (left), 0.5 (center), 0.9 (right).}
		\label{fig:plotrhoE}
		\end{figure}
			
    	\subsubsection{Double-scaling limit}
	        We have shown that the free energy, and consequently the phase structure and the critical curve, of the matrix model with symmetric weight of $E$ type at large $K$ and large $N$ agrees with the equivalent descriptions as a Meixner ensemble with a hard wall studied in \cite{CP:2013}. In fact, we find an even stronger agreement between the two pictures, as also the correlations among eigenvalues near the edge of the distributions in a double-scaling limit match.  For the unitary ensemble, the double-scaling limit is described in \cite{Baik:2000}, along the lines of  \cite{BDJ:1999}. On the other hand, in the discrete ensemble, if the critical region is approached from the small coupling phase, the hard wall is not active and the results of \cite{Johansson:2000} directly apply. The double-scaling is the same in both cases, and leads to the Tracy--Widom law \cite{TracyWidom:1992}. A third consistency check comes from the asymptotics of the hypergeometric kernel: standard computations using a steepest descend method  (see for example \cite[Sec. 4]{OLect} for an introduction) show that the result in the double-scaling limit matches the Tracy--Widom behaviour.

		\subsection{Phase transition: general case}
		\label{sec:Baikgen}
			While in the previous Section \ref{sec:Baiksym} we focused on the analysis of two matrix models that are symmetric under $z \leftrightarrow z^{-1}$, we now allow the more general situation $\beta_1 \ne \beta_2$, in which the symmetry is lost.\par
			The matrix models we consider here are
			\begin{align}
				\ZguH (t) &= \frac{1}{K!} \int_{[-\pi,\pi]^K} \frac{ \dd ^K \varphi }{(2 \pi)^K} \prod_{1 \le j < k \le K} \lvert e^{\ii \varphi_j} - e^{\ii \varphi_k} \rvert \prod_{j=1} ^K \left( 1-t e^{\ii \varphi_j} \right)^{- \beta_1} \left( 1-t e^{- \ii \varphi_j} \right)^{-\beta_2 } \label{eq:ZgenBaikH},	\\
				\ZguE (t) &= \frac{1}{K!} \int_{[-\pi,\pi]^K} \frac{ \dd ^K \varphi }{(2 \pi)^K} \prod_{1 \le j < k \le K} \lvert e^{\ii \varphi_j} - e^{\ii \varphi_k} \rvert \prod_{j=1} ^K \left( 1+t e^{\ii \varphi_j} \right)^{\beta_1} \left( 1+t e^{- \ii \varphi_j} \right)^{\beta_2 } \label{eq:ZgenBaikE} .
			\end{align}
			The weight functions are complex valued, thus we expect the eigenvalue densities at large $K$ to be complex-valued functions. We rewrite $(\beta_1,\beta_2)$ in terms of the parameters $(\beta, v)$ defined in \eqref{eq:defbetaandv}, and consider the limit $K \to \infty$ with scaling of $\beta$ as introduced in \eqref{eq:defgamma}.
			
			\subsubsection{The $H$-model}
				We first focus on the behaviour of the partition function \eqref{eq:ZgenBaikH} in the limit described above, and show how the discussion of the symmetric case in Section \ref{sec:Baiksym} is modified. Details of the calculations can be retrieved in Appendices \ref{app:ZgenHweak} and \ref{app:ZgenHstrong} for weak and strong coupling respectively.\par
				The leading contributions at large $K$ may be encoded in the eigenvalue density $\rho_H$ which solves the integral equation
				\begin{equation}
				\label{eq:SPEgenBaikH}
					- \ii \gamma t \left[ \frac{ (1-v) e^{\ii \varphi} }{1-t e^{\ii \varphi} } - \frac{ (1+v) e^{- \ii \varphi} }{1-t e^{- \ii \varphi} } \right] = \mathrm{P} \int \dd \vartheta \rho_{H} (\vartheta) \cot \left( \frac{ \varphi - \vartheta }{2} \right) .
				\end{equation}
				The asymmetry parameter $v$ complexifies the left hand side of the latter equation, and the resulting eigenvalue density is complex (see Appendices \ref{app:ZgenHweak} and \ref{app:ZgenHstrong}):
				\begin{equation*}
					\rho_H (\varphi )= \begin{cases} \frac{1}{2\pi} \left\{ 1 + 2 \gamma \sum_{n=1} ^{\infty} t^n \left[ \cos (n \varphi) - \ii v \sin (n \varphi) \right] \right\} , &  \gamma \le \gamma_{c,H} ,  \\ \frac{2 t  ( \gamma -1)  }{\pi (1-t) } \left[   \frac{  (1-t) \cos \left( \frac{ \varphi }{2} \right) - \ii v (1+t) \sin \left( \frac{ \varphi }{2} \right) }{  (1-t)^2 + 4 t \left( \sin \left( \frac{ \varphi }{2} \right) \right)^2 } \right] \sqrt{ \left( \sin \left( \frac{ \phi_0 }{2} \right) \right)^2 - \left( \sin \left( \frac{ \varphi }{2} \right) \right)^2 } , &  \gamma > \gamma_{c,H} , \end{cases} 
				\end{equation*}
				with critical value $\gamma_{c,H} = \frac{1+t}{2t}$, the same as in the symmetric case. The corresponding free energy $\mfgenuH = - K^{-2} \log \ZguH$, computed in Appendix \ref{app:FEgenboth}, can be written as:
				\begin{equation}
				\label{eq:FEgenH}
			  		\mfgenuH = \mfsymH + v^2 \Delta \mathcal{F} _H , \qquad \Delta \mathcal{F} _H  = \begin{cases}  - \gamma^2  \log (1-t^2) , &  \gamma \le \gamma_{c,H} (t) ,  \\ \log (1+t)  - \frac{1}{2} \log t, &  \gamma > \gamma_{c,H} (t) . \end{cases} 
				\end{equation}\par
				The introduction of the asymmetry reduces the order of the phase transition from third to second, with:
				\begin{equation*}
					\lim_{\gamma \uparrow \gamma_{c,H} } \frac{ \dd^2 \mfgenuH }{ \dd t^2} - \lim_{\gamma \downarrow \gamma_{c,H} } \frac{ \dd^2 \mfgenuH }{ \dd t^2} = - \frac{v^2}{t^2 (1-t)} ,
				\end{equation*}
				
			\subsubsection{The $E$-model}
				Consider now the second matrix model, the $E$-model of eq. \eqref{eq:ZgenBaikE}, and take its scaled limit as in \eqref{eq:defbetaandv}-\eqref{eq:defgamma}. The saddle point equation reads
				\begin{equation}
				\label{eq:SPEgenBaikE}
					- \ii \gamma t  \left( (1-v) \frac{ e^{\ii \varphi} }{1+t e^{\ii \varphi}} -  \frac{ (1+v) e^{ - \ii \varphi} }{1+t e^{- \ii \varphi}} \right) = \mathrm{P} \int \dd \vartheta \rho _{E} (\vartheta)  \cot \frac{ \varphi - \vartheta }{2}  .
				\end{equation}
				The eigenvalue density is complexified by the presence of the asymmetry parameter $v$ (see Appendices \ref{app:ZgenEweak} and \ref{app:ZgenEstrong}):
				\begin{equation*}
					\rho_E (\varphi )= \begin{cases} \frac{1}{2 \pi} \left[1 + 2 \gamma t \left( \frac{ \cos (\varphi ) - \ii v \sin (\varphi ) }{ 1 + t^2 + 2 t \cos (\varphi)  } \right) \right] , & \gamma \le \gamma_{c,E} , \\ \frac{ 2 t (\gamma+1) }{\pi (1+t) } \left[ \frac{ (1+t) \cos \left( \frac{\varphi}{2} \right) - \ii v (1-t) \sin  \left( \frac{\varphi}{2} \right) }{ (1+t)^2 -4t  \left( \sin\left( \frac{\varphi}{2} \right)\right)^2 } \right] \sqrt{ \left( \sin \left( \frac{\phi_0}{2} \right) \right)^2 -  \left( \sin \left( \frac{\varphi}{2} \right) \right)^2  } , & \gamma > \gamma_{c,E} , \end{cases}
				\end{equation*}
				with the same critical value as for the symmetric case: $\gamma_{c,E} = \frac{1-t}{2t}$. The free energy is (see Appendix \ref{app:FEgenboth}):
				\begin{equation}
				\label{eq:FEgenE}
				  \mfgenuE = \mfsymE + v^2 \Delta \mathcal{F} _E , \qquad \Delta \mathcal{F} _E  =  \begin{cases}  - \gamma^2 \log (1-t^2) , \quad \gamma \le \gamma_{c,E} (t) , \\ \log (1-t)  - \frac{1}{2} \log t ,  \quad  \gamma > \gamma_{c,E} (t) .  \end{cases}  
				\end{equation}\par
				The first derivative of the free energy is continuous, but the second derivative is not:
				\begin{equation*}
					\lim_{\gamma \uparrow \gamma_{c,E} } \frac{ \dd^2 \mfsymE }{ \dd t^2} - \lim_{\gamma \downarrow \gamma_{c,E} } \frac{ \dd^2 \mfsymE }{ \dd t^2} = - \frac{v^2}{t^2 (1+t)} ,
				\end{equation*}
				thus the phase transition is second order. We plot the first derivative of free energy as a function of $\gamma$, at different values of $t$, in Figure \ref{fig:3plotsfprime}.
				
				\begin{figure}[bth]
					\centering
					\includegraphics[width=0.45\textwidth]{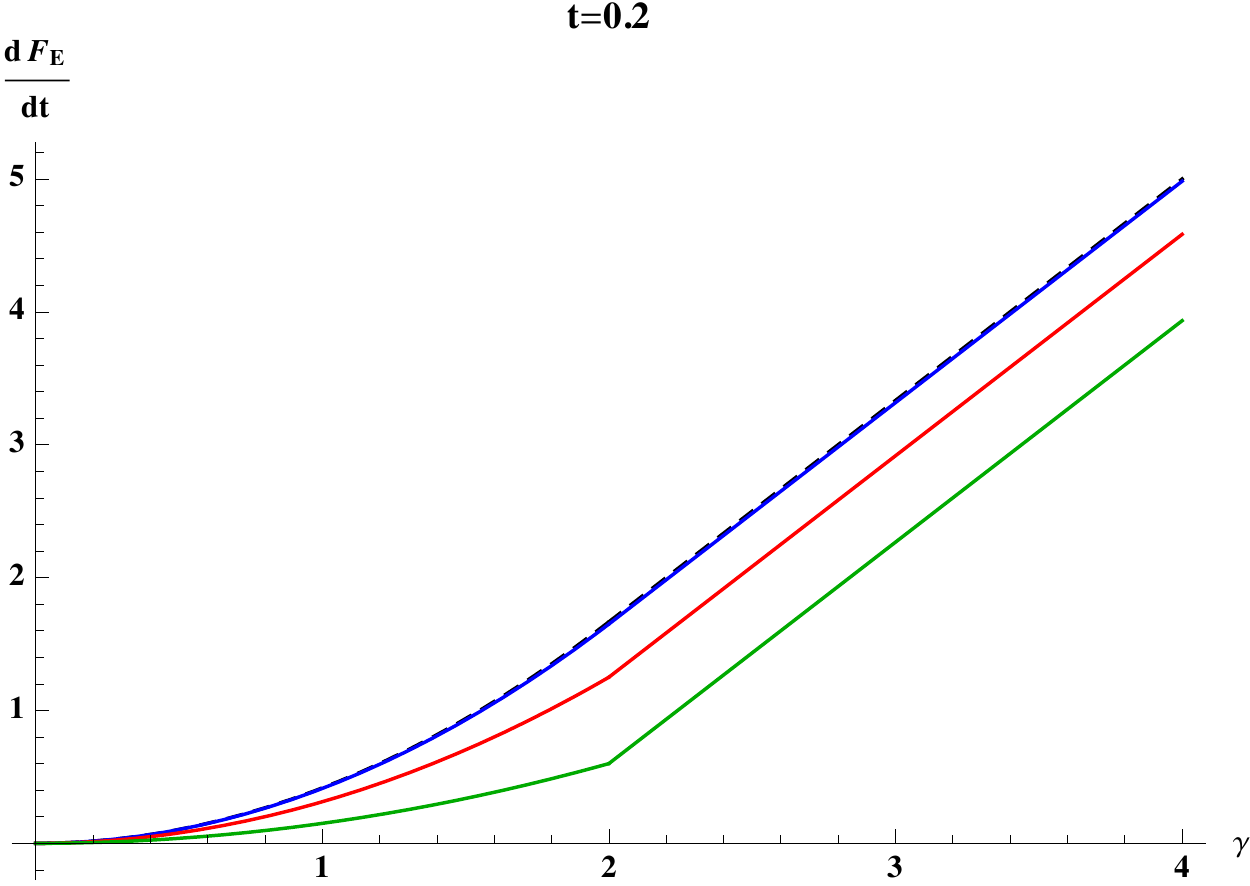}%
					\hspace{0.08\textwidth}
					\includegraphics[width=0.45\textwidth]{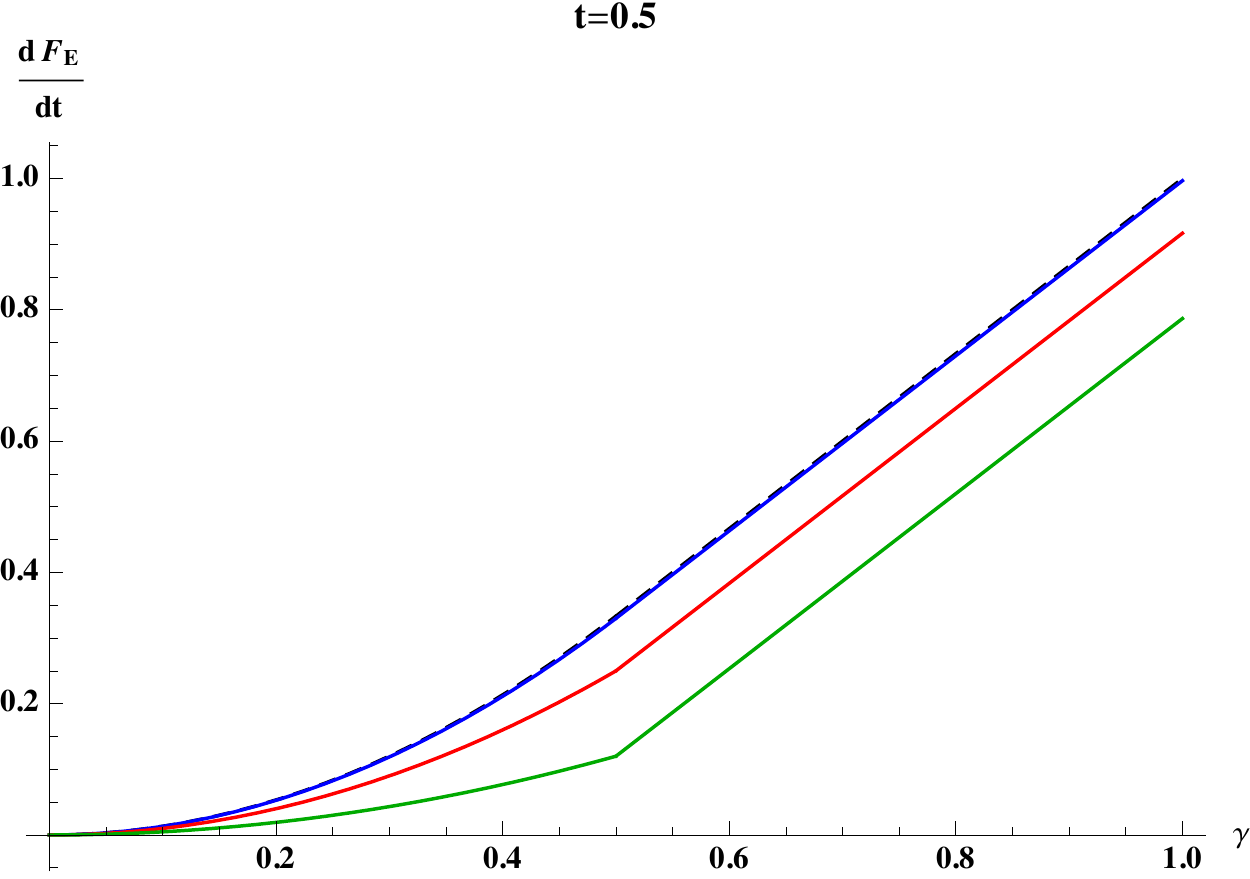}\\
					\vspace{0.5cm}
					\includegraphics[width=0.45\textwidth]{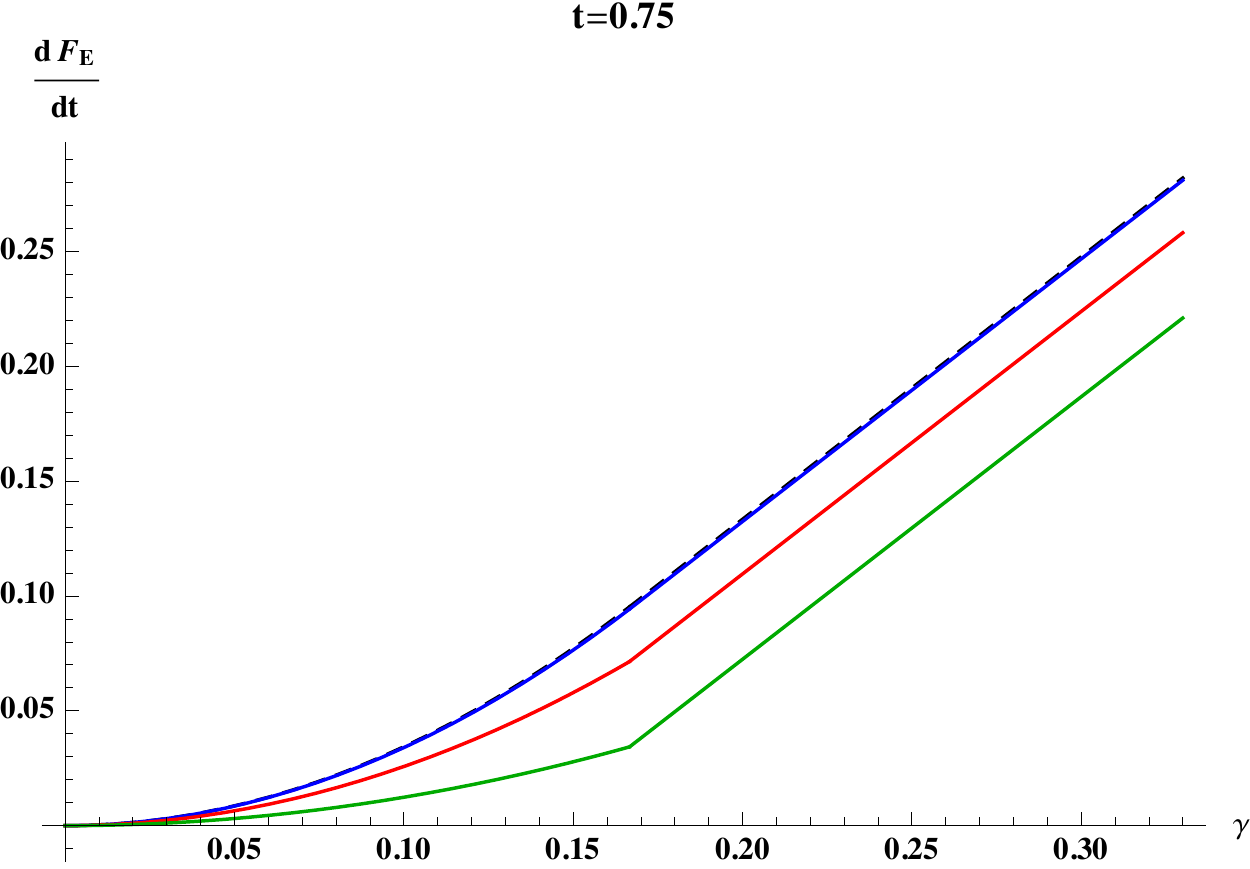}
					\caption{Planar limit of $\frac{ \dd \mathcal{F}_E }{\dd t} $ as a function of $\gamma$, for $t=0.2$ (up left), $t=0.5$ (up right), $t=0.75$ (down). In each plot, there appear $\frac{ \dd \mathcal{F}_E }{\dd t} $ for different values of the asymmetry parameter: $v=0.1$ (blue), $v=0.5$ (red), and $v=0.75$ (green). The dashed back line is the symmetric case $v=0$. The point at which the curve is continuous but with discontinuous derivative becomes more and more visible as $v$ is increased.}
									\label{fig:3plotsfprime}
				\end{figure}

	\subsection{The Gross--Witten limit}
	\label{sec:GWlimit}
		Consider the two matrix models of Section \ref{sec:Baiksym}, defined in equations \eqref{eq:BaikH} and \eqref{eq:BaikE}. The potential of those models can be written as:
		\begin{equation*}
			V_{E / H} ^{\mathrm{sym.}} (z) = \pm \beta \left[ \log \left( 1 \pm t z \right) +  \log \left( 1\pm t z^{-1} \right) \right] , 
		\end{equation*}
		with $+$ sign for the $E$ and $-$ for the $H$. We send $t \to 0$ and $\beta \to \infty$, keeping their product $\beta_{\mathrm{GW}} := t \beta$ fixed. This gives:
		\begin{equation*}
			V^{\mathrm{sym.}} _{E / H} (z) \to \beta_{\mathrm{GW}} \left( z + z^{-1} \right) ,
		\end{equation*}
		which is the potential of the Gross--Witten matrix model \cite{GW80,Wadia}. Note that the limit is the same for both the $E$- and the $H$-model. In the more general, non-symmetric case, with potential
		\begin{equation*}
			V_{E / H} ^{\mathrm{gen.}} (z) = \pm  \left[ \beta_1 \log \left( 1 \pm t z \right) +  \beta_2 \log \left( 1 \pm t z^{-1} \right) \right] ,
		\end{equation*}
		we pass from $(\beta_1, \beta_2)$ to $(\beta , v)$ as prescribed in \eqref{eq:defbetaandv}, and define $\beta_{\mathrm{GW}} := t \beta $. The same limit as above gives:
		\begin{equation*}
			V_{E / H} ^{\mathrm{gen.}} (z) \to \beta_{\mathrm{GW}} \left[  \left( z + z^{-1} \right) - v  \left( z - z^{-1} \right) \right],
		\end{equation*}
		which is the potential of the Gross--Witten model with a topological $\theta$-term \cite{Hisakado:1996a,Hisakado:1996b}. Since, by construction, we are in the regime $\lvert v \rvert < 1$, the Gross--Witten phase transition is still present, as discussed in \cite{Hisakado:1996b}.\par
		Within the setting laid down in the Introduction, the limit above with $\beta = N \in \mathbb{N}$ of the normalization of the $z$-measure leads to the normalization of the Poissonized Plancherel measure on partitions \cite{BOO,BOreview}. This is consistent with what we mentioned above, as the Poissonized Plancherel measure is mapped to the Gross--Witten unitary matrix model \cite{Johansson:1998}. On the other side, the Poissonized Plancherel measure is the $t \to 1$ limit of the Meixner ensemble \cite{Johansson:1999d} as well.

\subsection{Phase transition: from third to second order}
\label{sec:smoothdef}

For $0<\vert v \vert <1$, the potential of the unitary matrix models \eqref{eq:ZgenBaikH}-\eqref{eq:ZgenBaikE} becomes complex-valued. We have studied the large $K$ limit and showed that the model undergoes a second order phase transition, which becomes third order turning off the asymmetry parameter, $v \to 0$. Nevertheless, the phase transition in the Gross--Witten matrix model remains third order if the potential is modified into
\begin{equation*}
    \cos \varphi - \ii v \sin \varphi .
\end{equation*}

An alternative approach to face the matrix models with complexified potential is to analytically continue them in the sense of \cite{Witten:2010}. We relax the condition $\vert z \vert^2 =1$ and deform the integration cycle $\mathbb{S}^1 \rightsquigarrow \mathcal{C}_v$ in the complex plane so that the potential remains real along $\mathcal{C}_v$ (see Figure \ref{fig:def-contour}). Of course, for the deformation to be smooth, $\mathcal{C}_v$ should be a smooth Jordan curve, whose shape depends on $v$ and which is the unit circle $\mathbb{S}^1$ when $v=0$.\par

\begin{figure}[bth]
			\centering	\includegraphics[width=0.5\textwidth]{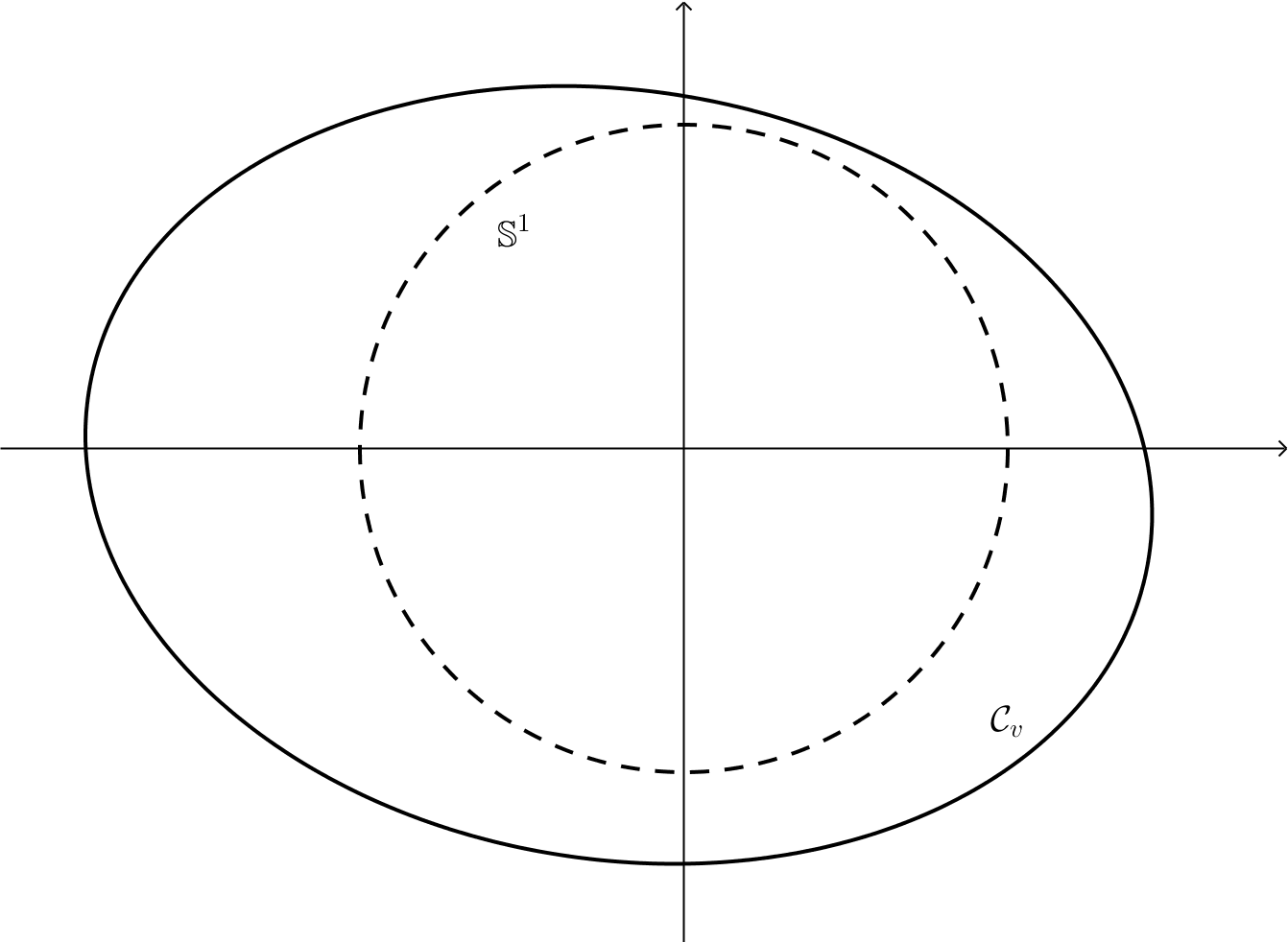}%
	\caption{We deform the integration contour from the unit circle to a Jordan curve along which the potential is real.}
		\label{fig:def-contour}
		\end{figure}

As a warm up, we apply the analytic continuation to the Gross--Witten model: we seek a contour $\mathcal{C}_{v}$ such that
\begin{equation*}
     \left[ \left( z + z^{-1} \right)  - v \left( z - z^{-1} \right) \right] \in \mathbb{R}
\end{equation*}
for all $z \in \mathcal{C}_v$. We find that $\mathcal{C}_v$ is simply a circle of radius $\sqrt{(1+v)/(1-v)}$. Moreover, the potential is 
\begin{equation*}
    V (z) = 2 \sqrt{1-v^2} \cos \varphi, \qquad  z \in \mathcal{C}_v, \ \varphi=\mathrm{Arg}~z .
\end{equation*}
Hence, not only it is real-valued, but we also recover the original, symmetric Gross--Witten model sitting on a rescaled circle, see Figure \ref{fig:GW-def-contour}. 
Furthermore, $\mathcal{C}_v$ is sent to infinity or shrinks to a point as $v \to 1$ or $v \to -1$ respectively, and we cannot prolong beyond these values. This provides a new perspective on the result of \cite{Hisakado:1996a,Hisakado:1996b}, where a drastic difference in the behaviour was observed crossing from $0<|v|<1$ to $|v|>1$.

\begin{figure}[bth]
			\centering	\includegraphics[width=0.35\textwidth]{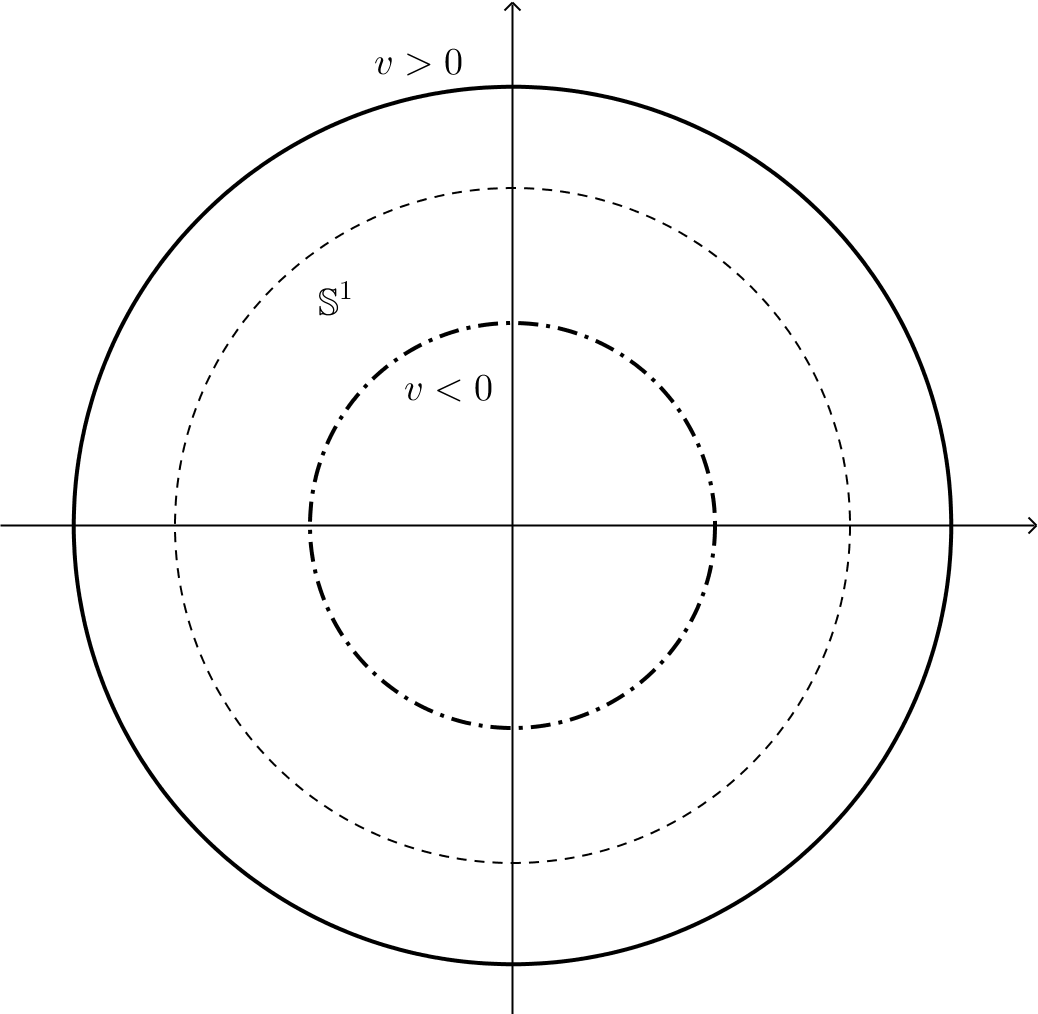}%
	\caption{The deformed contour for the complexified Gross--Witten model is a new circle, with bigger ($0<v<1$) or smaller ($-1<v<0$) radius.}
		\label{fig:GW-def-contour}
		\end{figure}

For a generic Laurent polynomial potential 
\begin{equation*}
    (1-v)\sum_{k=1} ^{n} c_k z^k +  (1+v)\sum_{k=1} ^{n} c_k z^{-k} , \qquad c_k \in \mathbb{R} , \ k=1, \dots , n ,
\end{equation*}
the smooth integration cycle $\mathcal{C}_v$ is determined as the locus in $\mathbb{C}$ that solves the algebraic equation
\begin{equation*}
    \Im \left( z^n V(z) \right) =0 , \qquad \forall z \in \mathcal{C}_v .
\end{equation*}
Stated more formally, provided the potential is a Laurent polynomial, we determine a suitable contour $\mathcal{C}_v$ as a rational algebraic curve of genus zero. If the curve has disconnected components, we simply retain as integration cycle the component homeomorphic to $\mathbb{S}^1$. Under the assumption that the coefficients $\left\{ c_k \right\}_{k=1, \dots, n}$ are generic, this guarantees that the locus $\mathcal{C}_v$ will be a small deformation of $\mathbb{S}^1$ if $V(z)$ is deformed by a small $v \ne 0$. See, for example, \cite{AlgGeom} for a textbook reference on affine algebraic curves.\par
However, the potentials of the matrix models with weights of $E$ and $H$ type are not polynomials but logarithms. The form of the potential is precisely the reason for the discontinuous behaviour as a function of $\vert N_1 - N_2 \vert $, because then imposing
\begin{equation}
\label{eq:condImV0}
	\Im V_{E / H} ^{\mathrm{gen.}} (z) = 0 , \qquad  z \in \mathcal{C}_v,
\end{equation}
does not define an algebraic curve embedded in $\mathbb{C}$. 
For every real $v \ne 0$ we find that the unique contour solving \eqref{eq:condImV0} is a half real line together with an open segment (the details depend on whether we consider the $E$- or the $H$-model), which is not a smooth deformation of $\mathbb{S}^1$.\par
\medskip
It is perhaps instructive to look at the problem from the converse perspective. In general, the existence of a Jordan curve $\mathcal{C}_{\underline{v}} \subset \mathbb{C}$, depending on a collection of parameters $\underline{v}$, along which the potential is real-valued is not guaranteed. For non-polynomial potential $V(z)$, it requires to place the model at certain special points of the parameter space. For the weights of $E$ and $H$ type, this dictates $v=0$, equivalently $N_1 =N_2$.\par

\subsubsection{Phase transition and universality}

To sum up our conclusions in one sentence, the unitary matrix models with symmetric $E$ and $H$ weight undergo a third order GW phase transition, but their more general, non-symmetric extensions have complex potentials and therefore fall out of the GW universality class. 
This ought to be contrasted with what happens in the equivalent description as a discrete ensemble with a hard wall. In that case, increasing the parameter $\gamma$ from zero, the system undergoes a third order phase transitions at the value $\gamma = \gamma_{c,E}$ when the hard wall becomes active. This holds both in the symmetric and the more general setting: they belong to the same universality class. See \cite{Cunden:2017,Cunden:2018} for detailed discussion on the universality of the third order phase transition in presence of a hard wall. We also stress that the discrete topology further constrains the eigenvalue density, but the condition is satisfied for all values of the parameters \cite{CP:2013,CP:2015}. Thus, the discrete nature of the ensemble plays no role in determining the phase structure of this model.\par
It is worth mentioning that the symmetric $H$-model has also been studied in the context of supersymmetric gauge theories, both with unitary \cite{ChenMekU} and orthogonal and symplectic symmetry \cite{ChenMekOSp}, and the third order phase transition was observed also in the latter case \cite{ChenMekPT}\footnote{The unitary matrix model of \cite[Eq. (2.18)]{ChenMekU} corresponds to a generic choice $(t,t^{\prime})$, thus the weight function is complex. Setting $t^{\prime} = t$ in \cite{ChenMekU} combines chiral and anti-chiral matter fields: it would be interesting to interpret the change of order of the phase transition from second to third in that context.}.

\section{Random matrix ensembles on the real line}

We now discuss two random matrix models on the real line that follow from a
change of variables of the unitary matrix models above described. One family
of random matrix ensembles is obtained through the 1d stereographic
projection, leading to a Cauchy--Romanovski type of ensembles and the other
one is based on the mapping $x=\cos \theta $ which is more useful when the
original matrix model has a symplectic or orthogonal symmetry (that is, corresponding with Toeplitz$\pm $Hankel determinants, instead of Toeplitz determinants), but
that also lead to explicit expressions for the symmetric unitary matrix
models, in terms of Wronskians of Chebyshev polynomials of the four types.

\subsection{Cauchy ensembles and Romanovski orthogonal polynomials}

\label{sec:stereoBaik}

We analyse the matrix models $\ZUsymH$ and $\ZUsymE$ defined in eq. %
\eqref{eq:BaikH} and \eqref{eq:BaikE} and consider a change of topology, in
which we remove a single point from the unit circle, to obtain $\mathbb{S}%
^{1}\setminus \left\{ \infty \right\} \cong \mathbb{R}$. We expect the asymptotic behavior of these models at large $K
$ to correspond to the ``gapped'' phase of the unitary matrix model (see Figure \ref{fig:gapPT}), which appears at strong coupling. We pass from angular variables to the real line using the stereographic
projection \cite[Sec. 2.5]{Forrester}: 
\begin{equation*}
e^{\ii\varphi }=\frac{1+\ii x}{1-\ii x},\qquad -\pi <\varphi <\pi ,\ x\in \R,
\end{equation*}%
as sketched in Figure \ref{fig:stereoproj}. 
\begin{figure}[bth]
			\centering	\includegraphics[width=0.6\textwidth]{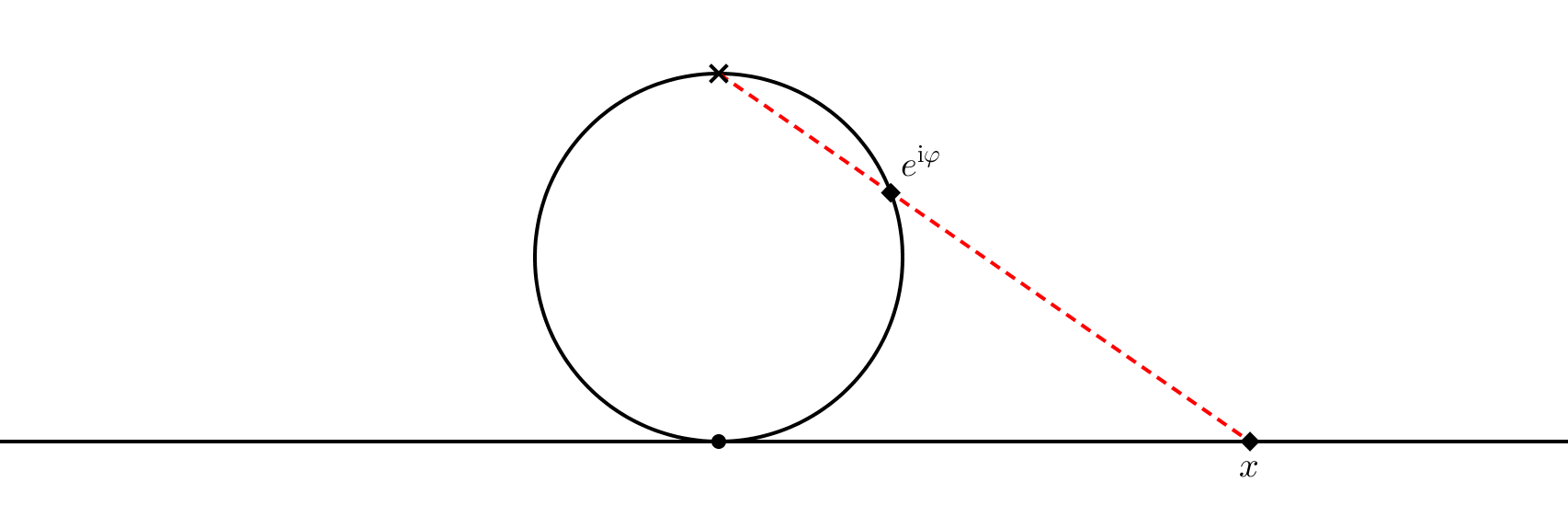}%
	\caption{The 1d stereographic projection.}
		\label{fig:stereoproj}
		\end{figure}

The Vandermonde determinant is mapped to: 
\begin{equation*}
\prod_{1\leq j<k\leq K}\left\vert e^{\ii\varphi _{j}}-e^{\ii\varphi
_{k}}\right\vert ^{2}\dd\varphi _{1}\dots \dd\varphi
_{K}=2^{K^{2}}\prod_{1\leq j<k\leq K}\left( x_{j}-x_{k}\right)
^{2}\prod_{j=1}^{K}\frac{1}{(1+x_{j}^{2})^{K}}\dd x_{1}\dots \dd x_{K},
\end{equation*}%
and, with the corresponding transformation of the symbols, 
the resulting matrix models are:
\begin{align}
	\mathcal{Z}_{H,\mathrm{stereo.}} ^{\mathrm{sym.}} &= \frac{ 2^{K^2} }{K! (2 \pi)^K } \int_{\R^K} \dd ^K x \prod_{1 \le j < k \le K} \left( x_j - x_k \right)^2  \prod_{j=1} ^K \frac{(1+x_j^2 )^{\beta-K}}{ \left[ (1-t)^2 + x_j^2 (1+t)^2 \right]^{\beta} } , \label{eq:stereosymH}\\
	\mathcal{Z}_{E,\mathrm{stereo.}} ^{\mathrm{sym.}} &= \frac{ 2^{K^2} }{K! (2 \pi)^K } \int_{\R^K} \dd ^K x \prod_{1 \le j < k \le K} \left( x_j - x_k \right)^2  \prod_{j=1} ^K \frac{ \left[ (1+t)^2 + x_j^2 (1-t)^2 \right]^{\beta} }{(1+x_j^2 )^{\beta+K}} . \label{eq:stereosymE}
\end{align}\par
Likewise, we can also use the stereographic projection in the more general,
non-symmetric, matrix models \eqref{eq:ZgenBaikH} and %
\eqref{eq:ZgenBaikE}. We obtain:
\begin{align}
	\mathcal{Z}_{H,\mathrm{stereo.}} ^{\mathrm{gen.}} &= \frac{ 2^{K^2} }{K! (2 \pi)^K } \int_{\R^K} \dd ^K x \prod_{1 \le j < k \le K} \left( x_j - x_k \right)^2  \prod_{j=1} ^K (1+x_j^2 )^{\beta -K}  \notag \\
		& \ \times \left[ \left( \frac{1}{  (1-t)^2 + x_j^2 (1+t)^2  } \right) \left( \frac{ (1-t) - 2 \ii t x_j + x_j ^2 (1+t) }{ (1-t) + 2 \ii t x_j + x_j ^2 (1+t) } \right)^v \right]^{\beta}, \label{eq:stereogenH} \\
	\mathcal{Z}_{E,\mathrm{stereo.}} ^{\mathrm{gen.}} &= \frac{ 2^{K^2} }{K! (2 \pi)^K } \int_{\R^K} \dd ^K x \prod_{1 \le j < k \le K} \left( x_j - x_k \right)^2  \prod_{j=1} ^K \frac{1}{(1+x_j^2 )^{\beta +K} } \notag \\
		& \ \times \left[  \left( (1+t)^2 + x_j^2 (1-t)^2 \right) \left( \frac{ (1+t) - 2 \ii t x_j + x_j ^2 (1-t) }{ (1+t) + 2 \ii t x_j + x_j ^2 (1-t) } \right)^v \right]^{\beta} . \label{eq:stereogenE}
\end{align}
where we used the redefinition of parameters $\beta _{1}=\beta (1-v)$ and $\beta _{2}=\beta (1+v)$.

\subsubsection{Exact evaluation}

\label{sec:exactevZEreal} In the limit case $t=1$, the partition function of
the symmetric $E$-model, after stereographic projection, takes the simple
form: 
\begin{equation}
\mathcal{Z}_{E,\mathrm{stereo.}}^{\mathrm{sym.}}(t=1)=\frac{2^{K^{2}+2K\beta
}}{K!(2\pi )^{K}}\int_{\R^{K}}\dd^{K}x\prod_{1\leq j<k\leq K}\left(
x_{j}-x_{k}\right) ^{2}\prod_{j=1}^{K}\left( 1+x_{j} ^2\right) ^{-K-\beta }.
\label{eq:Romanovskiens}
\end{equation}%
This random matrix ensemble has been studied as a Cauchy ensemble \cite{FWCUE}, Lorentz ensemble \cite{Brouwer}, and in this form it is a
particular case of the classical ensemble with weight \cite{Forrester,Koepf}
\begin{equation*}
\sigma (x)=(1-\ii x)^{-\alpha _{1}}(1+ \ii x)^{-\alpha _{2}},\qquad \alpha
_{1}+\alpha _{2}>1,\ x\in \R.
\end{equation*}%
This weight function satisfies the Pearson equation and hence
the associated random matrix ensemble is classical \cite{Koepf,HuertaMarcellan,Forrester},
although in many references the listing of classical ensembles appears restricted to Hermite, Laguerre and Jacobi. See \emph{e.g.}  \cite{HuertaMarcellan} for the expanded list of possible classical weights. The associated polynomials go under
many names, including pseudo-Jacobi \cite{Lesky,Leskybook,JordaanTookos,SongWong,Wongreview} due to their
(non-trivial) relationship with Jacobi polynomials, and also appear in \cite{Aldaya}. A
proper name seems Romanovski polynomials $\left\{ R_{n}^{(\alpha _{1},\alpha
_{2})}\right\} $ given \cite{Romanovski,Lesky}, see \cite{Raposo2007} for a review. We evaluate now the
matrix integral \eqref{eq:Romanovskiens} using Romanovski polynomials. They satisfy (the dependence on $(\alpha_1, \alpha_2)$ is understood): 
\begin{equation*}
\frac{1}{2\pi }\int_{\R}\dd x R_{m} R_{n}(1-\ii x)^{-\alpha _{1}}(1+\ii %
x)^{-\alpha _{2}}=h_{n}\delta _{mn}
\end{equation*}%
where we have chosen a normalization such that the polynomials $R_{n}$ are
monic, and their norm squared is \cite{AskeyBeta} 
\begin{equation*}
h_{n}=2^{-\alpha _{1}-\alpha _{2}+1}\left[ (-1)^{n+1}\frac{\Gamma (-n+\alpha
_{1}+\alpha _{2})}{n!(2n-\alpha _{1}-\alpha _{2}+1)\Gamma (-n+\alpha
_{1})\Gamma (-n+\alpha _{2})}\right] \left[ \frac{n! 2^{n} \Gamma (-2n+\alpha
_{1}+\alpha _{2})}{\Gamma (-n+\alpha _{1}+\alpha _{2})}\right] ^{2}
\end{equation*}%
where the second square bracket is the change of normalization to obtain
monic Romanovski polynomials from the normalization of \cite{AskeyBeta}. Therefore,
using 
\begin{equation*}
\mathcal{Z}_{E,\mathrm{stereo.}}^{\mathrm{sym.}}(t=1)= \left( \frac{2^{K^{2}+2K\beta
}}{K!} \right) \ K!  \prod_{n=0}^{K-1}h_{n}  ,
\end{equation*}%
with $h_{n}$ specialized to the case of interest $\alpha _{1}=\alpha
_{2}=K+\beta $, we obtain 
\begin{equation*}
\mathcal{Z}_{E,\mathrm{stereo.}}^{\mathrm{sym.}}(t=1) =G(K+1)\frac{G(\beta +1)^{2}G(K+2\beta +1)}{G(K+\beta +1)^{2}G(2\beta +1)},
\end{equation*}%
where the Barnes $G$-function \cite{barnes1899theory} is identified using%
\begin{equation*}
\prod_{n=0}^{K-1}\Gamma (K+\beta -n)=\frac{G(K+\beta +1)}{G(\beta +1)}.
\end{equation*}%
The result is the same indeed as eq. \eqref{eq:exactZEt1}. Therefore, the $%
\mathcal{Z}_{E}(t=1)$ is computed in three different ways: as a discrete
ensemble on a finite set, as a unitary ensemble, and as the Cauchy ensemble
on the real line. The tools used in each of the three approaches are,
respectively: the Hahn polynomials \cite{CP:2015}, the Toeplitz determinant
with symbol with a pure Fisher--Hartwig singularity \cite{BS:1985}, and the
Romanovski polynomials, as represented in Figure \ref{fig:diagt1}.

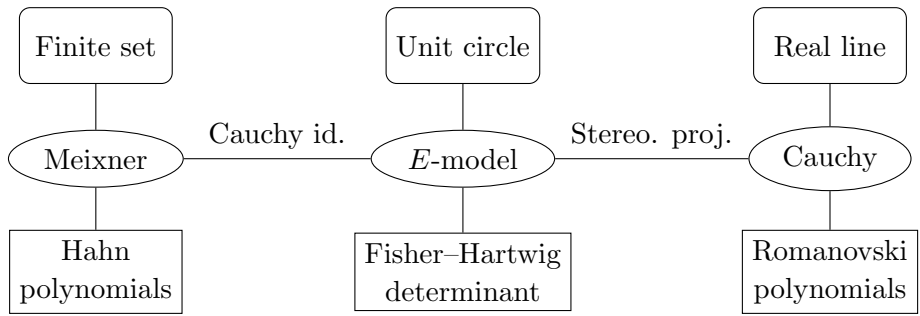
\begin{figure}[bth]
\centering
\begin{tikzpicture}[auto,node distance=2.8cm]
	\node[entity, rounded corners] (unit)  {Unit circle}
	child{node[attribute] (Ew) {$E$-model}
	child{node[entity, align=center] {Fisher--Hartwig \\ determinant }}
	};
	\node[entity, rounded corners] (discrete) [left = of unit] {Finite set}
	child{node[attribute] (Meix) {Meixner}
	child{node[entity, align=center] {Hahn \\ polynomials}}
	};
	\node[entity, rounded corners]  (stereo) [right = of unit] {Real line}
	child{node[attribute] (Cauchy) {Cauchy}
	child{node[entity, align=center] {Romanovski \\ polynomials}}
	};
	\path (Meix) edge node  {Cauchy id.}  (Ew);
	\path (Ew) edge   node {Stereo. proj.} (Cauchy);
\end{tikzpicture}
\caption{Relationships between matrix ensembles: the supports (above), the weight functions (middle) and the corresponding tools providing an exact solution at $t=1$ (below).}
\label{fig:diagt1}
\end{figure}

\subsubsection{Asymmetric Toeplitz matrix}

A much studied symbol in spectral analysis of Toeplitz matrices \cite{ekstrom2018eigenvalues,Bottcheretal} is%
\begin{equation}
\label{eq:symbolomono}
\sigma \left( z\right) =z^{s}(1+tz)^{\beta }(1+tz^{-1})^{\beta},
\end{equation}
for integer $s\in \mathbb{Z}$. As discussed in Section \ref{sec:unitexact}, this has the effect of shifting the Fourier
coefficients of the symbol by $s$, thus it shifts the diagonals of the
Toeplitz matrix upward (if $s>0$) or downward ($s<0$) by $s$ and the
Toeplitz banded matrix associated to our model becomes asymmetric, making
it more similar to the generalized model in terms of this asymmetry of the
associated Toeplitz matrix. The stereographic projection of this extra term
is: 
\begin{equation*}
e^{\ii s\varphi }=\left( \frac{1+\ii x}{1-\ii x}\right) ^{s}=\exp \left(
s\log \frac{1+\ii x}{1-\ii x}\right) =e^{2s\arctan x}.
\end{equation*}%
For $t=1$, the resulting weight function is
\begin{equation}
\label{eq:shiftsymbEstereo}
    \sigma (x) = (1- \ii x)^{-K - \beta -s} (1+ \ii x)^{-K - \beta +s} .
\end{equation}
This is the stereographic projection of the weight introduced in eq. \eqref{unitaryshiftsymb}, and it is still of the Romanovski form (see \cite{Koepf}) with the identification%
\begin{equation*}
\alpha _{1}=K+\beta +s,\quad \alpha _{2}=K+\beta -s.
\end{equation*}
We can therefore give explicit evaluation of the determinant of the Toeplitz banded matrix with symbol \eqref{eq:shiftsymbEstereo} using again Romanovksi polynomials: 
\begin{align}
\det T_{K} \left( \sigma \left( x \right) \right) &=\left(\frac{K!}{2^{K^2-2\beta K}}  \right) G(K+1) \frac{ G(\beta + s +1) G(\beta - s +1)  }{ G(K +\beta + s +1) G(K+ \beta - s +1) } \frac{G (K+ 2 \beta +1) }{G (2 \beta +1)} \label{detshiftedstereo}.
\end{align}
Note that, if we consider the random matrix ensemble with weight \eqref{eq:shiftsymbEstereo} and keep the normalization such that it coincides with the stereographic projection of the unitary ensemble with weight \eqref{unitaryshiftsymb}, discussed in Section \ref{sec:unitexact}, the factor in bracket in \eqref{detshiftedstereo} cancels exactly.

\subsection{Large $K$ limit}
\label{sec:largeKromanov}

We now focus on $\mathcal{Z}_{E,\mathrm{stereo.}} ^{\mathrm{sym.}}$, in eq. %
\eqref{eq:stereosymE}, and study its large $K$ limit, with $\gamma= \beta/K$
fixed. The saddle point equation is: 
\begin{equation}  \label{eq:stereoSPEsym}
\mathrm{P} \int \dd y \frac{ \rho (y) }{x-y} = (1+ \gamma) \frac{x}{1+x^2} -
\gamma \frac{ (1-t)^2 x }{ (1+t)^2 + x^2 (1-t)^2 } .
\end{equation}
The parity symmetry of the matrix model guarantees that we can look for a
symmetric solution with $\mathrm{supp} \rho = [-A, A]$, $A>0$. We expect
that $A$ will be back-projected to $e^{\ii \phi_0}$ of Section \ref{sec:Baiksym} undoing the stereographic projection.

We report the details of the
computations in Appendix \ref{app:ZEstereosym}. We arrive to the
eigenvalue density: 
\begin{equation*}
\rho (x) = \frac{ \sqrt{A^2 - x^2 }}{\pi} \left[ \frac{ (1+ \gamma ) }{ 
\sqrt{ A^2 + 1 } (x^2 + 1) } - \frac{ \gamma t_0 }{ \sqrt{ A^2 + t_0 ^2 }
(x^2 + t_0 ^2 ) } \right] ,
\end{equation*}
with the boundary $A$ fixed by normalization: 
\begin{equation*}
A^2 = \frac{ (2 \gamma + 1 )(1+t)^2 }{ (2 \gamma + 1 - t )( 2 \gamma - 1 +
t) } .
\end{equation*}
This solution is well defined as long as 
\begin{equation*}
\gamma > \frac{1-t}{2t} ,
\end{equation*}
and sending $A \to \infty$, which would be back-projected to $\phi_0 \to \pi$%
, corresponds to the limit $\gamma \downarrow \gamma_{c,E} = \frac{1-t}{2t} $%
. This matches the analysis on the circle.

We can do the same with the more general matrix model $\mathcal{Z}_{E,%
\mathrm{stereo.}}^{\mathrm{gen.}}$ of eq. \eqref{eq:stereogenE}. In this
case, as happened on the circle, turning on the asymmetry parameter $v$
complexifies the eigenvalue density, which becomes: 
\begin{equation*}
\rho _{\mathrm{stereo.}}(x)=\frac{\sqrt{A^{2}-x^{2}}}{\pi }\left[ \frac{%
(1+\gamma )+\ii v\gamma x}{\sqrt{A^{2}+1}(x^{2}+1)}-\frac{\gamma t_{0}+\ii %
v\gamma x}{\sqrt{A^{2}+t_{0}^{2}}(x^{2}+t_{0}^{2})}\right] ,
\end{equation*}%
with same value of $A$ as above. The derivation of the result is given in
Appendix \ref{app:ZEstereogen}. For completeness, we also report the details of the large $K$ analysis of the matrix model with weight \eqref{eq:symbolomono}, corresponding to the determinant of an asymmetric Toeplitz matrix, in Appendix \ref{app:stereo+mono}.

\subsection{Orthogonal and symplectic symmetries and Wronskians of Chebyshev polynomials}
\label{sec:Wronski}

The other option, to pass from the unit circle to the real line, is to use $%
x=\cos \theta $. This works best for matrix integrals where the Haar
measure is that of $O(2K)$, $Sp(2K)$ or $O(2K+1)$. That is, for
Toeplitz$\pm$Hankel determinants. In that case the result  \cite[Lemma 2.7]{DIK} holds, connecting Toeplitz$\pm$Hankel and Hankel determinants.
\begin{lemma}[\cite{DIK}]
\label{lema27DIK}Let $\sigma_{j}$ be the Fourier coefficient of the symbol $\sigma (z)$, $\sigma_{j}=%
\frac{1}{2\pi }\int_{- \pi}^{\pi } \sigma (e^{\ii \theta })e^{-\ii j \theta }d\theta $, with\footnote{The symmetric symbols of $E$ and $H$ type satisfy the assumption, but the non-symmetric symbols do not.} $\sigma(e^{\ii \theta })=\sigma(e^{-\ii \theta })$. Then,
\begin{eqnarray*}
\frac{1}{2} \det {\left( \sigma_{j-k}+\sigma_{j+k-2}\right) _{j,k=1}^{K}} &{=}&\frac{2^{K^{2}-K+1}}{%
(2\pi) ^{K}}D_{K}\left( \sigma (e^{\ii \theta \left( x\right) })/\sqrt{1-x^{2}} \right) \\
\det {\left( \sigma_{j-k}-\sigma_{j+k}\right) _{j,k=1}^{K}} &{=}&\frac{2^{K^{2}+K}}{%
(2\pi) ^{K}}D_{K}\left( \sigma (e^{\ii \theta \left( x\right) })\sqrt{1-x^{2}} \right) \\
\det {\left( \sigma_{j-k}+\sigma_{j+k-1}\right) _{j,k=1}^{K}} &{=}&\frac{2^{K^{2}}}{%
(2\pi) ^{K}}D_{K} \left( \sigma (e^{\ii \theta \left( x\right) })\sqrt{\frac{1+x}{1-x}} \right) \\
\det {\left( \sigma _{j-k}- \sigma _{j+k-1}\right) _{j,k=1}^{K}} &{=}&\frac{2^{K^{2}}}{%
(2\pi) ^{K}}D_{K} \left( \sigma (e^{\ii \theta \left( x\right) })\sqrt{\frac{1-x}{1+x}} \right)
\end{eqnarray*}%
where $D_{K}( \cdot )$ is the Hankel determinant on $[-1,1]$, with $x=\cos \theta $. 
\end{lemma}
We give a new proof of this result following \cite{GGT}. First, we set the notation for the rest of the Section. Denote $O^{\pm } (N)$ the real orthogonal matrices with determinant $ \pm 1$, and
\begin{equation*}
    G (K) \in \left\{ O^{+}(2K), Sp(2K),  O^{+} (2K+1), O^{-}(2K+1) \right\} .
\end{equation*}
Note that the Haar measures on $Sp (2K)$ and $O^{-} (2K+2)$ coincide. Let us consider matrix integrals with weight $\sigma (e^{\ii \theta})$ over the group $G(K)$ (see \cite[Sec. 2.6]{Forrester} for details on the Haar measure on orthogonal and symplectic groups), which we denote by $\mathcal{Z}_{\sigma} ^{G(K)}$.
\begin{proof}[Proof of Lemma \ref{lema27DIK}]
    The change of variable $x_j = \cos \theta_j$ gives
    \begin{equation}
        \mathcal{Z}_{\sigma} ^{G(K)} = \frac{ \mathcal{N} _{G(K)} }{K! (2 \pi)^K} \int_{\left[ -1,1\right] ^{K}}%
\prod_{j=1}^{K} \sigma (e^{\ii \theta (x)}) \left( 1+x_{j}\right) ^{a}\left( 1-x_{j}\right) ^{b}\dd x_{j}\prod_{1\leq j<k\leq
K}\left( x_{j}-x_{k}\right) ^{2} \label{ZsigmaGK}
    \end{equation}
    with $a, b \in \left\{ \pm 1/2 \right\} $ depending on the symmetry of the original model,
\begin{equation}
	\begin{tabular}[bht]{l|c | c | c | c }
	    \ & $O^{+} (2K)$ &  $Sp (2K)$ & $O^{+} (2K+1) $  & $O^{-} (2K+1) $ \\
	    \hline
	    $a$ & $-\frac{1}{2}$ & $ \frac{1}{2}$ & $-\frac{1}{2}$ & $ \ $ $ \frac{1}{2}$ \\
	    $b$ & $-\frac{1}{2}$ & $\frac{1}{2}$ & $ \ $ $ \frac{1}{2}$ & $-\frac{1}{2}$
	\end{tabular}
\label{abZOSp}
\end{equation}
    and the overall normalization depends on the specific $G(K)$:
\begin{equation*}
    \mathcal{N} _{O^{+} (2K) } = 2^{K^2 -K +1} , \quad  \mathcal{N} _{Sp (2K) } = 2^{K^2 +K } , \quad  \mathcal{N} _{O^{\pm} (2K+1) } = 2^{K^2 }  .
\end{equation*}
    Using Andr\'{e}ief's identity \cite{Andreief,MeetA}, expression \eqref{ZsigmaGK} gives the four Hankel determinants on the right hand sides in the Lemma, according to the choice of $G(K)$. The Toeplitz$\pm$Hankel determinants on the left hand sides in the Lemma correspond with matrix integration over orthogonal and symplectic group \cite{BaikRains1,GGT}. Indeed, for $\sigma (e^{\ii \theta}) = \sigma (e^{-\ii \theta}) $
\begin{align*}
\int_{O^{+} (2N)} \sigma (U)dU & =\frac{1}{2}\det {\left( \sigma_{j-k}+\sigma_{j+k-2}\right)
_{j,k=1}^{N}}, \\
\int_{Sp(2N)} \sigma (U)dU & =\det {\left( \sigma_{j-k}-\sigma_{j+k}\right) _{j,k=1}^{N}}.
\end{align*}
Furthermore, as noted in \cite{GGT},
\begin{equation*}
    \int_{G(K)} \sigma (U)dU= \int_{G(K)}\sigma (-U)dU , \qquad \text{  for  } G(K) \in \left\{ U(K),O^{+}(2K), Sp (2K) \right\} ,
\end{equation*}
but then we have 
\begin{align*}
\int_{O^{+}(2K+1)} \sigma (U) dU=\det {\left( \sigma_{j-k}-\sigma_{j+k-1}\right) _{j,k=1}^{K}}, \\
\int_{O^{+}(2K+1)} \sigma (-U)dU= \int_{O^{-}(2K+1)} \sigma (U)dU=\det {\left( \sigma_{j-k}+\sigma_{j+k-1}\right) _{j,k=1}^{K}} .
\end{align*}
This concludes the proof of the four identities.
\end{proof}

For the case of the symmetric $E$ and $H$ symbols considered in the paper, we have expressions of the
general form:
\begin{equation}
\mathcal{Z}_{E/ H} ^{G(K)} = \frac{ \mathcal{N}_{G(K)} }{K! (2 \pi)^K} ( 2t)^{\pm \beta K }\int_{\left[ -1,1\right] ^{K}}%
\prod_{j=1}^{K}\left( \xi \pm x_{j}\right) ^{\pm \beta }\left(
1+x_{j}\right) ^{a}\left( 1-x_{j}\right) ^{b}\dd x_{j}\prod_{1\leq j<k\leq
K}\left( x_{j}-x_{k}\right) ^{2},  \label{eq:EorHOSp}
\end{equation}
with $\xi =\left( 1+t^{2}\right) /2t$ and $a,b$ as in \eqref{abZOSp}. The $\pm$ signs are always the upper sign for the $E$-model and the lower for the $H$-model. Moreover, we take $\beta=N \in \mathbb{Z}_{>0}$ a positive integer.
Thus, $\mathcal{Z}_{E/H} ^{G(K)}$ are averages of moments of characteristic polynomials, taken in a random matrix ensemble where the weight function is of the
Jacobi type, $\omega \left( x\right) =\left( 1+x\right) ^{a}$ $\left(
1-x\right) ^{b}$, but always with the specific values of $a$ and $b$ above,
corresponding to the Chebyshev subfamily. Hence, the corresponding orthogonal
polynomials for 
\begin{equation*}
a=b=-\frac{1}{2}, \qquad  a=b=\frac{1}{2}, \qquad a=-b=\frac{1}{2}, \qquad a=-b=-\frac{1}{2}
\end{equation*}
are
Chebyshev polynomials of first, second, third and fourth kind respectively, denoted by 
\begin{equation*}
T_{n}\left( x\right) , \quad \qquad \quad U_{n}(x), \quad \qquad \quad V_{n}(x), \quad \qquad \quad W_{n}(x).
\end{equation*}
Their correspondence with group integration is \eqref{abZOSp}. As mentioned, instead of interpreting the resulting model \eqref{eq:EorHOSp} as the
partition function of a model with a semi-classical weight function $%
\widetilde{\omega }\left( x \right) =\left( \xi- x_{j}\right) ^{\pm \beta
}\left( 1+x\right) ^{a}$ $\left( 1-x\right) ^{b}$ (which in itself is an
interesting possibility and may lead to connections with Painlev\'{e}
equations and other integrable models), we adopt the characteristic polynomial
point of view \cite{BrezinHikami,BDStrahov}, which, for the $+$ sign case in \eqref{eq:EorHOSp} (below we discuss the other choice), gives that%
\begin{equation*}
\frac{ \mathcal{Z}_{E}^{G(K)} (\beta = N ) }{ \mathcal{Z}_{E}^{G(K)}\left( \beta =0\right) }= (-2t)^{N K} C_{E} (N) \ \mathrm{Wr}\left[ P_{K
}(- \xi ),...,P_{K+N-1}( - \xi )\right] ,
\end{equation*}%
where $\mathrm{Wr}\left[ f_{n}(x),...,f_{n+l-1}(x)\right] \equiv \det \left(
f_{n+k-1}^{(j-1)}\left( x\right) \right) _{j,k=1}^{l}$ denotes the Wronskian
determinant and $P_{n}(\cdot )$ denotes one of the four Chebyshev polynomials. We also have defined the coefficients
\begin{equation*}
    C_{E} (N) = \frac{ (-1)^{\frac{N (N-1)}{2}} }{G(N+1)} , \qquad C_H (N) = \left(\frac{4}{\pi}\right)^N \frac{1}{G(N+1)}.
\end{equation*}
($C_H$ is for later use). The result follows immediately from Lemma \ref{lema27DIK} above (\emph{i.e.} \cite[Lemma 2.7]{DIK}) together with the known characteristic polynomial results \cite{BrezinHikami}. It has
also been the specific subject of a later paper \cite{Elouafi1}, where the same result is obtained, for a polynomial and, as above, even, symbol. Therefore it holds
that:

\begin{eqnarray*}
\left( \frac{1}{(-2t)^{N K} C_{E} (N) } \right)\ \frac{ \mathcal{Z} _{E}^{O^{+}(2K)} (\beta =N )  }{ \mathcal{Z}_{E}^{O^{+}(2K)}\left( \beta =0\right) } &=&  \mathrm{Wr}\left[
T_{K }( - \xi ),...,T_{K+N-1}( - \xi )\right] , \\
\left( \frac{1}{(-2t)^{N K} C_{E} (N) } \right) \ \frac{ \mathcal{Z} _{E}^{Sp(2K)} (\beta =N )  }{ \mathcal{Z}_{E}^{Sp(2K)}\left( \beta =0\right) } &=& \mathrm{Wr} \left[
U_{K }( - \xi ),...,U_{K+N-1}( - \xi )\right] , \\
\left( \frac{1}{(-2t)^{N K} C_{E} (N) } \right) \  \frac{\mathcal{Z}_{E}^{O^{-}(2K+1)} (\beta=N )  }{\mathcal{Z}_{E}^{O^{-}(2K+1)}\left( \beta
=0\right) } &=& \mathrm{Wr} \left[
V_{K }( - \xi ),...,V_{K+N-1}( - \xi )\right] , \\
\left( \frac{1}{(-2t)^{N K} C_{E} (N) }\right) \ \frac{ \mathcal{Z}_{E}^{O^{+} (2K+1)} (\beta =N ) }{ \mathcal{Z}_{E}^{O^{+}(2K+1)}\left( \beta =0\right) } &=& \mathrm{Wr} \left[ W_{K }( - \xi ),...,W_{ K+N-1}( - \xi ) \right] .
\end{eqnarray*}%

In addition, we can also study the case with symbol $H$ in the same manner.
For negative moments the characteristic polynomial
description becomes involved, having to consider not the polynomials
orthogonal w.r.t. the given weight, but the Cauchy transform of
these polynomials \cite{BDStrahov}. In some cases this is a
considerable complication, however, for Chebyshev polynomials, it is well known that the polynomials of first and second kind (and likewise the polynomials of third and fourth kind) are integral transforms of each other with respect
to weighted Hilbert kernels \cite{Mason}:%
\begin{align*}
 \mathrm{P} \int_{-1}^{1}\frac{T_{n}(y) \mathrm{d} y}{(x-y)\sqrt{1-y^{2}}} &= \pi U_{n-1}(x)  \ \text{ and }  \mathrm{P}\int_{-1}^{1}\frac{U_{n}(y)\sqrt{1-y^{2}}\mathrm{d}y}{(x-y)} = -\pi T_{n+1}(y), \\
 \mathrm{P} \int_{-1}^{1}\frac{V_{n}(y)}{(x-y)}\sqrt{\frac{1+y}{1-y}} \mathrm{d} y &= \pi W_{n}(x)  \quad \text{ and }  \mathrm{P} \int_{-1}^{1}\frac{W_{n}(y)}{(x-y)}\sqrt{\frac{1-y}{1+y}} \mathrm{d} y =  -\pi
V_{n}(x),
\end{align*}%
with $\mathrm{P} \int$ the Cauchy principal value integral. The weighting is precisely the corresponding weight function, which is how it is required \cite{BDStrahov}. Note also that $\xi >1$ for $0<t<1$, thus $\pm \xi$ falls outside the integration region: in this case the Hilbert and Cauchy transform simply differ by a factor $\sqrt{-1} /2$. Therefore, when $0\le N <K$ we have:  
\begin{eqnarray*}
\left( \frac{ (2t)^{N K } }{C_H (N)} \right) \ \frac{ \mathcal{Z}_{H}^{O^{+}(2K)} (\beta =N )  }{  \mathcal{Z}_{H}^{O^{+}(2K)}\left( \beta =0\right) } &=&  (-1)^{N} \ \mathrm{Wr} \left[
U_{K-N -1}(\xi ),...,U_{K-2}(\xi )\right] , \\
\left( \frac{ (2t)^{N K } }{C_H (N)}\right) \  \frac{\mathcal{Z}_{H}^{Sp(2K)} (\beta=N )  }{ \mathcal{Z}_{H}^{Sp(2K)}\left( \beta =0\right) } &=&  \mathrm{Wr} \left[
T_{K-N +1}(\xi ),...,T_{K}(\xi )\right] , \\
\left( \frac{ (2t)^{N K } }{C_H (N)} \right) \ \frac{\mathcal{Z}_{H}^{O^{-} (2K+1)} (\beta=N )  }{\mathcal{Z} _{H}^{O^{-} (2K+1)}\left( \beta
=0\right) } &=& (-1)^N  \ \mathrm{Wr} \left[
W_{K-N }(\xi ),...,W_{K-1}(\xi )\right] , \\
\left( \frac{ (2t)^{N K } }{C_H (N)} \right) \ \frac{ \mathcal{Z}_{H}^{O^{+} (2K+1)} (\beta=N ) }{ \mathcal{Z}_{H}^{O^{+} (2K+1)}\left( \beta =0\right) } &=&   \mathrm{Wr} \left[ V_{K-N }(\xi ),...,V_{K-1}(\xi ) \right] .
\end{eqnarray*}%
Thus, we obtain a duality between the $E$ model with $O^{-}(2K+1)$ and the $H$ model with $O^{+} (2K+1)$ and conversely. The parameters are mapped $K \mapsto K-N$, $\xi \to -\xi$ under the duality. We also obtain a tight relation between the $E$ model with $O^{+} (2K) $ and the $H$ model with $Sp (2K)$ and conversely, although it is not really a duality in the sense that we cannot recover one result from the other via a bijective map of the parameters.

So far, we have discussed the same models as above but with a symplectic or
orthogonal symmetry (it would be interesting if such model had an
interpretation in terms of a dimer problem, perhaps with boundaries, but we
do not know of such an interpretation). The unitary case follows from the
identity \cite{VeinDale}:%
\begin{eqnarray*}
\det T_{2K+1}( \sigma ) &=&\frac{1}{2}\det {\left( \sigma_{j-k}- \sigma_{j+k}\right)
_{j,k=1}^{K}}\det {\left( \sigma_{j-k}+ \sigma_{j+k-2}\right) _{j,k=1}^{K+1} } , \\
\det T_{2K}(\sigma ) &=&\det {\left( \sigma _{j-k}+ \sigma_{j+k-1}\right) _{j,k=1}^{K}}\det {
\left( \sigma_{j-k}-\sigma_{j+k-1}\right) _{j,k=1}^{K}} ,
\end{eqnarray*}
which is used in matrix integral form in \cite{GGT} and was used in \cite{Elouafi2} to precisely extend the Wronskian result \cite{Elouafi1} (again for a symbol of the $E$ type) to the unitary case. Therefore, from this
Wronskian point of view, the unitary matrix model is more complicated than
the orthogonal and symplectic ones, because it is the product of two Wronskians. The formulas for the symmetric $E$ symbol are\footnote{Note that in the odd case the two Wronskians correspond to $Sp (2K)$ and $O^{+} (2K +2)$, not $O^{+} (2K)$.}:
\begin{align*}
   \widetilde{C}_{E}^{U(2K+1)} (N) \ \frac{ \mathcal{Z}_E ^{U(2K+1)} (\beta=N )  }{ \mathcal{Z}_E ^{U(2K+1)} (\beta=0)} & = \mathrm{Wr}\left[
T_{K +1}( - \xi ),...,T_{K+N}( - \xi )\right] \cdot \mathrm{Wr}\left[
U_{K }( - \xi ),...,U_{K+N-1}( - \xi )\right] , \\ 
\widetilde{C}_{E}^{U(2K)} (N) \ \frac{ \mathcal{Z}_E ^{U(2K)} (\beta =N ) }{ \mathcal{Z}_E ^{U(2K)} (\beta=0)} & = \mathrm{Wr}\left[
V_{K }( - \xi ),...,V_{K+N-1}( - \xi )\right] \cdot \mathrm{Wr}\left[
W_{K }( - \xi ),...,W_{K+N-1}( - \xi )\right] ,
\end{align*}
and for the $H$ symbol
\begin{align*}
 \widetilde{C}_{H} ^{U(2K+1)} (N) \ \frac{ \mathcal{Z}_H ^{U(2K+1)} (\beta =N)  }{ \mathcal{Z}_H ^{U(2K+1)} (\beta=0)}  & = (-1)^N \ \mathrm{Wr}\left[
T_{K-N +1}( \xi ),...,T_{K}( \xi )\right] \cdot \mathrm{Wr}\left[
U_{K-N -1 }( \xi ),...,U_{K-2}( \xi )\right] , \\ 
\widetilde{C}_{H} ^{U(2K)} (N) \ \frac{ \mathcal{Z}_H ^{U(2K)} (\beta =N ) }{ \mathcal{Z}_H ^{U(2K)} (\beta=0)}  & = (-1)^N \ \mathrm{Wr}\left[
V_{K-N }( \xi ),...,V_{K-1}( \xi )\right] \cdot \mathrm{Wr}\left[
W_{K-N }( \xi ),...,W_{K-1}( \xi )\right] ,
\end{align*}
with coefficients
\begin{equation*}
    \widetilde{C}_{E} ^{U(M)} (N) = \frac{G(N+1)^2 }{(-2t)^{N M} } , \quad \widetilde{C}_{H} ^{U(M)} (N) =  G(N+1)^2 (2t)^{N M} \left( \frac{\pi}{4} \right)^{2N} , \qquad M \in \left\{ 2K , 2K+1 \right\} .
\end{equation*}
We obtain a duality for $U(2K)$ with $E$ and $H$ symbol, with parameters mapped according to $\xi \mapsto - \xi$ and $K \mapsto K-N$ together with the factor $(-1)^N$ in the $H$-model. The relation between the model with $E$ and $H$ symbol and group $U(2K+1)$ is not a duality.

It is worth mentioning that Wronskians of orthogonal polynomials in general
satisfy an interesting identity which expresses them as  Hankel determinants,
but with the roles of the $K$ and $N$ swapped \cite{Lec}. In the case
of the Chebyshev Wronskian of the first type, the Hankel determinant is in
terms of Legendre polynomials and for the second type one is in terms of a
Hankel of Gegenbauer polynomials $C^{(3/2)} _n (x)$ (polynomials orthogonal
w.r.t. $ \omega \left( x\right) =(1-x^{2})$).

\subsubsection{Large $K$ limit}
The asymptotics of a Hankel determinant with symbol $ f(x) (1+x)^{a} (1-x)^{b}$, $x \in [-1,1]$ with $f(x) >0$ and $f^{\prime} (x)$ Lipschitz continuous, has been obtained by Basor and Chen \cite{BasorChen} when $a \ge 0, b\ge 0$. They also showed that the resulting expression can be analytically continued to all $a\ge -\frac{1}{2}, b \ge -\frac{1}{2}$, and therefore it applies to the matrix ensemble \eqref{eq:EorHOSp} for each of the four symmetry groups. 
In that case, $a,b \in \left\{ + \frac{1}{2}, - \frac{1}{2} \right\}$ and $f(x) = (\xi \pm x)^{\pm \beta}$, with upper sign for the $E$ model and lower sign for the $H$ model. Recall also that $\xi = (1+t^2)/(2t)>1$, hence $f(x)$ is admissible in both models. The formula of \cite{BasorChen} for generic $a,b$ and $f$ is:
\begin{equation}
    D_K \left( f(x) (1+x)^{a} (1-x)^{b} \right) \approx  \frac{(2 \pi)^{K} K ^{\frac{a^2 + b^2}{2} - \frac{1}{4} } }{2^{K (K+ a +b)}} \exp \left( \frac{K}{\pi} \int_{-1} ^{1} \frac{ \log f(x) }{ \sqrt{1-x^2}} \mathrm{d} x  \right) Z (f) ,
\label{eq:BCformula}
\end{equation}
for large $K$, where $Z (f)$ is independent of $K$ and is given by:
\begin{align*}
    Z (f) & = \exp \left[ \int_{-1} ^{1} \frac{ \log f(x) }{ \sqrt{1-x^2}} \left( \mathrm{P} \int_{-1} ^{1} \frac{ f^{\prime} (y) \sqrt{ 1-y^2}}{ f(y) (y-x) } \frac{ \mathrm{d} y }{2 \pi} \right)\frac{ \mathrm{d} x }{2 \pi} \right] \\
    & \times \exp \left( \frac{a+b}{2\pi} \int_{-1} ^{1} \frac{ \log f(x) }{ \sqrt{1-x^2}} \mathrm{d} x  \right) \frac{ G \left( \frac{ a+b+1}{2} \right)^2 G \left( \frac{ a+b+2}{2} \right)^2 }{ G \left(  a+b+1 \right)  G \left(  a+1 \right)  G \left(  b+1 \right) } \Gamma \left( \frac{ a+b+1}{2} \right) .
\end{align*}
This expression directly gives the asymptotic behaviour of $\mathcal{Z}^{G(K)} _{E/H}$ for both symbols. 
First, we notice that the numerical prefactor in \eqref{eq:EorHOSp} simplifies with the one coming from formula \eqref{eq:BCformula}, exactly for $G(K)\in \left\{ Sp(2K), O^{\pm} (2K+1) \right\}$ and leaving behind a factor of $2$ when $G(K) = O^{+} (2K)$. 
Then, we use
\begin{equation*}
    \int_{-1} ^{1} \frac{ \log (\xi \pm x) }{ \sqrt{1-x^2}} \frac{ \mathrm{d} x }{ \pi }  = \log \left( \frac{ \xi + \sqrt{ \xi^2 -1}}{2} \right) = \log \frac{1}{2t} 
\end{equation*}
and see that the contribution from the exponential in \eqref{eq:BCformula} cancels:
\begin{equation*}
    (2t)^{\pm \beta K}  \exp \left( \pm \beta K    \int_{-1} ^{1} \frac{ \log (\xi \pm x) }{ \sqrt{1-x^2}} \frac{ \mathrm{d} x }{ \pi } \right) = 1 .
\end{equation*}

We are left with the contribution from $Z(f)$. The ratio of products of Barnes $G$-functions and the Gamma function is trivial for $G(K)\in \left\{ Sp(2K), O^{\pm} (2K+1) \right\}$, and is $\frac{1}{2}$ for $G(K) = O^{+} (2K)$, thus cancelling the left over factor from the numerical coefficient. The exponential of the single integral gives
\begin{equation*}
 \exp \left( \pm (a+b) \beta  \int_{-1} ^{1} \frac{ \log (\xi \pm x) }{ \sqrt{1-x^2}}  \frac{\mathrm{d} x}{2\pi}  \right) = (2 t )^{ \mp (a+b) \frac{ \beta}{2}}  ,
\end{equation*}
and in particular is trivial for odd rank symmetry, $G(K)=O^{\pm} (2K+1)$. The double integral which appears in the exponential in $Z(f)$ is universal for the four symmetries and takes the form
\begin{equation*}
   \beta^2 \int_{-1} ^{1} \frac{ \log (\xi \pm x) }{ \sqrt{1-x^2}}  \mathrm{P} \int_{-1} ^{1} \frac{  \sqrt{ 1-y^2}}{ (y-x) (y \pm \xi) } \frac{ \mathrm{d} y }{2 \pi} \frac{ \mathrm{d} x }{2 \pi} = \frac{ \beta^2}{4 \pi} \int_{-1} ^{1} \frac{ \log (\xi \pm x) }{ \sqrt{1-x^2}}  \left( -1 + \frac{ \sqrt{\xi^2 -1 }}{\xi \pm x} \right)  \mathrm{d} x .
\end{equation*}
The last integral admits an exact solution in terms of hypergeometric functions, albeit involved. Instead, we write $\xi$ explicitly as $(1+t^2)/2t$ and, using $x= \cos \theta$, we expand the logarithm:
\begin{align*}
    \log \left( \xi \pm x \right) & = - \log 2t + \log (1 \pm t e^{\mathrm{i} \theta} ) + \log (1 \pm t e^{-\mathrm{i} \theta} ) \\ & = - \log 2t - \sum_{n =1}^{\infty} \frac{ 2 (\mp t)^{n}}{n} \cos (n \theta ) =  - \log 2t - \sum_{n =1}^{\infty} \frac{ 2 (\mp t)^{n}}{n} T_n (x) .
\end{align*}
After integration the $n=0$ term cancels, while for $n\ge1$ we use the Hilbert transform property of the Chebyshev polynomials to arrive at the final expressions
\begin{align*}
    \mathcal{Z}^{O^{\pm} (2K+1)} _{E/H} & \approx \exp \left( \pm \frac{ \beta^2}{4} \left( 1 - t^2 \right) \sum_{n=0} ^{\infty} \frac{ (\mp t)^{n}}{n+1} U_{n} (\mp \xi)  \right) , \\
    \mathcal{Z}^{O^{+} (2K)} _{E/H} & \approx \exp \left( \pm \frac{ \beta^2}{4} \left( 1 - t^2 \right) \sum_{n=0} ^{\infty} \frac{ (\mp t)^{n}}{n+1} U_{n} (\mp \xi)  \pm \frac{\beta}{2} \log 2t \right) , \\
     \mathcal{Z}^{Sp (2K)} _{E/H} & \approx \exp \left( \pm \frac{ \beta^2}{4} \left( 1 - t^2 \right) \sum_{n=0} ^{\infty} \frac{ (\mp t)^{n}}{n+1} U_{n} (\mp \xi)  \mp \frac{\beta}{2} \log 2t \right) ,
\end{align*}
with upper sign for the $E$-model and lower sign for the $H$-model. Note that we cannot recast the series into a logarithm because $\xi >1$.\par
We obtained the large $K$ limit with fixed $\beta$, but we could also study the large $K$ and large $\beta$ limit with $\gamma = \beta/K$ fixed, as we did in Section \ref{sec:Baiksym} for the matrix models with unitary symmetry. This scaled limit was analyzed in \cite{ChenMekPT} (except for $O^{-} (2K+1)$, which is nonetheless easily found following their argument), and a third order phase transition from a gapless to a gapped phase was observed, in agreement with the results in the unitary setting.

\section{Outlook}

It would be interesting if semiclassical polynomials can be used to obtain further analytical results on the model with orthogonal and symplectic symmetries, using the Wronskian representation presented here. Notice that computing these Wronskians is equivalent to computing determinants of banded matrices of the Toeplitz$\pm$Hankel type, which arise for example as Laplacian or difference matrices of problems with boundaries \cite{strang2014functions,cunden2018free}.

Other analytical approaches could have been followed as well. For example, it is immediate to interpret \eqref{eq:stereosymE} as a characteristic polynomial average over a Cauchy ensemble\footnote{Note that because the arguments of the polynomials would be purely imaginary, the determinantal expressions would be in terms of Jacobi polynomials, for this model.}. It would be more interesting to use this approach for the case of the asymmetric symbol \eqref{eq:symbolomono}, which, when mapped on the real line, gives the full Cauchy--Romanovski weight. This could be studied in the context of asymmetrically banded Toeplitz matrices. 

At large $N$, it could be interesting to consider the recent results \cite{Charlier} on  correlators of characteristic polynomials for classical ensembles. 
If a confluent limit of these results on correlators can be taken, that would directly evaluate the Wronskian description of the orthogonal and symplectic ensembles discussed in Section \ref{sec:Wronski}. On the other hand, the asymptotic formulas in \cite{Charlier} already could be directly applied to a $q$-deformation of the $E$-model, such as the one studied in \cite[Sec. 4.3]{GGT1}.

It is worth mentioning that the Cauchy ensemble has interesting connections with the Laguerre ensemble and Painlev\'{e} equations \cite{winn,basor2019representation} that may be further exploited in our context. We could also focus further on the \emph{exact equivalences} part and see if other averages over the Meixner ensemble, such as the ones in \cite{cohen2019moments}, have a counterpart in terms of the $E$ and $H$ unitary random matrix ensembles. \par
\medskip

With regards to the phase structure, it seems also worth exploring if the phase transitions obtained, which are phase transitions of the determinant in the Toeplitz matrix representation of the $E$ and $H$ models, are related to any spectral change of the Toeplitz matrices, which in the case of the $E$-model is a banded matrix. The spectra of such matrices is still of very much interest and a rich subject \cite{ekstrom2018eigenvalues,Bottcheretal} and it could be a worthwhile direction to study the spectra in the natural scaling limit suggested by the random matrix representation, the one explored here. 

We conclude by mentioning that passing from a third order to a second order phase transition opens up the possibility of intriguing interpretations. According to \cite{Hanada:19}, the negative jump of the second derivative of the free energy when crossing the critical line is associated with the creation of a meta-stable state, and the second order phase transition corresponds to tunneling from the meta-stable vacuum to the stable one. This aspect may deserve further attention.

\subsection*{Acknowledgements}

The work of MT was partially supported by the Funda\c{c}\~{a}o para a Ci\^{e}ncia e a Tecnologia (FCT) 
through its program Investigador FCT IF2014, under contract 
IF/01767/2014. 
The work of LS is supported by the FCT through the doctoral grant SFRH/BD/129405/2017. The work is also supported by FCT Project PTDC/MAT-PUR/30234/2017.


\begin{appendix}

\section{Solution to saddle point equations on the unit circle}
\label{app:saddlecalc}

Here we present the calculations to solve the saddle point equations obtained in the scaled large size limit of the unitary matrix models discussed in Section \ref{sec:BaikMM}.

\subsection{Solution for $\ZUsymH$ at weak coupling}
	\label{app:ZUHweak}
	We start with the solution of the saddle point equation \eqref{SPEBaikH}, which we rewrite here for completeness:
	\begin{equation}
			\label{appeq:SPEBaikH}
				- \ii \gamma t \left[ \frac{ e^{\ii \varphi} }{1-t e^{\ii \varphi} } - \frac{ e^{- \ii \varphi} }{1-t e^{- \ii \varphi} } \right] = \mathrm{P} \int \dd \vartheta \rho_H (\vartheta) \cot \left( \frac{ \varphi - \vartheta }{2} \right) .
	\end{equation}
	The solution was originally obtained by Baik in \cite{Baik:2000}, and the present derivation via saddle point equations was discussed in \cite{ML:2019}.\par
	We make the ansatz $\mathrm{supp} \rho = [- \pi, \pi ]$. On the right hand side of \eqref{appeq:SPEBaikH} we use the expansion (for $\vartheta \ne \varphi$):
	\begin{equation*}
			\cot \left( \frac{ \varphi - \vartheta }{2} \right) = 2 \sum_{n=1} ^{\infty} \left[ \sin (n \varphi) \cos (n \vartheta) - \cos (n \varphi) \sin (n \vartheta ) \right],
	\end{equation*}
	thus the integration extracts the Fourier coefficients of the eigenvalue density:
	\begin{equation*}
		\mathrm{P} \int \dd \vartheta \rho (\vartheta) \cot \left( \frac{ \varphi - \vartheta }{2} \right) = 2 \pi \sum_{n=1} ^{\infty} \left[ a_n \sin (n \varphi) - b_n \cos (n \varphi) \right] ,
	\end{equation*}
	where
	\begin{equation*}
				a_n = \frac{1}{\pi} \int_{-\pi} ^{\pi} \dd \vartheta \rho_H (\vartheta ) \cos (n \vartheta) , \qquad b_n = \frac{1}{\pi} \int_{-\pi} ^{\pi} \dd \vartheta \rho_H (\vartheta ) \sin (n \vartheta) .
	\end{equation*}
	For the left hand side we use:
	\begin{equation*}
				- \ii \gamma t \left[ \frac{ e^{\ii \varphi} }{1-t e^{\ii \varphi} } - \frac{ e^{- \ii \varphi} }{1-t e^{- \ii \varphi} } \right] = - (\ii \gamma) \frac{ 2 \ii t \sin ( \varphi) }{1 + t^2 - 2t \cos (\varphi) } = 2 \gamma t \sin (\varphi) \sum_{n=0} ^{\infty} U_n ( \cos (\varphi)) t^n  ,
	\end{equation*}
	where in the second equality we recognized the generating function of Chebyshev polynomials of second kind, $U_n(x)$. Exploiting the basic property $U_n (\cos \varphi) = \frac{ \sin (n+1) \varphi }{ \sin \varphi}$, equation \eqref{appeq:SPEBaikH} is rewritten as:
	\begin{equation*}
				2 \gamma \sum_{n=1} ^{\infty} \sin (n \varphi) t^n = 2 \pi \sum_{n=1} ^{\infty} \left[ a_n \sin (n \varphi) - b_n \cos (n \varphi) \right] ,
	\end{equation*}
	and we immediately get:
	\begin{equation}
				a_n = \frac{ \gamma }{\pi} t^n , \qquad b_n = 0,
	\end{equation}
	for all $n=1,2, \dots $. The only yet undetermined coefficient $a_0$ is fixed by normalization to $a_0 = 1/(2 \pi)$. Putting all together:
	\begin{equation*}
				\rho_{H} (\varphi) = \frac{1}{2 \pi} \left[ 1 + 2 \gamma \sum_{n=1} ^{\infty} \cos (n \varphi) t^n \right] ,
	\end{equation*}
	which can be further simplified recognizing the generating function of Chebyshev polynomials of first kind $T_n (x)$, using the property $T_n (\cos (\varphi)) = \cos (n \varphi)$. We finally arrive to:
	\begin{equation}
			\label{eq:rhoBaikH1}
				\rho_{H} (\varphi) = \frac{1}{2 \pi} \left[ 1 + 2 \gamma t \left(  \frac{  \cos \varphi - t  }{1 + t^2 - 2 t \cos \varphi } \right) \right] .
	\end{equation}	
	The minima of this function are located at $\varphi = \pm \pi$, and imposing the condition $\rho (\varphi ) \ge 0$ for all $- \pi \le \varphi \le \pi$, we see that the present solution holds in the regime
	\begin{equation*}
		0 \le \gamma \le \frac{1+t}{2t} =: \gamma_{c,H} .
	\end{equation*}
	Above the critical value $\gamma_{c,H}$, solution \eqref{eq:rhoBaikH1} ceases to be valid and we must drop the assumption $\mathrm{supp} \rho = [- \pi, \pi ]$ and look for a different solution.

\subsection{Solution for $\ZUsymH$ at strong coupling}
		\label{app:ZUHstrong}
		In this Appendix we present the calculations to solve  the saddle point equation \eqref{SPEBaikH} at strong coupling. The procedure follows \cite[App. B]{MinwallaWadia:2013} (see also \cite[Sec. 4]{ML:2019}). We look for a one-cut solution with $\mathrm{supp} \rho = [- \phi_0, \phi_0]$, $0 < \phi_0 < \pi$.\par
		We introduce the following complex function, named resolvent:
		\begin{equation*}
				\Psi (z) = \int_{\mathcal{L}} \frac{ \dd u }{ \ii u} \frac{z+u}{z-u} \psi_H (u) , \qquad z \in \C  \setminus \mathcal{L}.
		\end{equation*}
		Here, $\mathcal{L} \subset \mathbb{S}^1$ is the arc of the unit circle from $e^{-\ii \phi_0} $ to $e^{\ii \phi_0} $, oriented anti-clockwise, and the complex function $\psi_H (u)$ is the continuation of $\rho_{H} (\varphi)$ to the complex plane, with $\psi_H (e^{\ii \varphi}) = \rho_{H} (\varphi)$. The eigenvalue density is recovered from the resolvent through the relation
		\begin{equation*}
			\Psi_{+} (e^{\ii \varphi}) - \Psi_{-} (e^{\ii \varphi}) = 4 \pi \psi_H (e^{\ii \varphi}) ,
		\end{equation*}
		with the subscript $+$ (resp. $-$) meaning the limit taken from outside (resp. inside) the unit circle.\par
		Following standard methods, we introduce two auxiliary complex functions:
		\begin{equation}
			 h(z) = \sqrt{ (e^{\ii \phi_0} -z) (e^{- \ii \phi_0} -z) } , \qquad z \in \C ,
		\end{equation}
		and a function $\Phi (z)$, to be determined, such that
		\begin{equation}
				\Psi (z) = h (z) \Phi (z)  , \qquad z \in \C .
		\end{equation}
		The saddle point equation \eqref{SPEBaikH} becomes a discontinuity equation for $\Phi (z)$. We will use this discontinuity equation, together with the fact that, by definition, $\Phi (z)$ decays as $\sim 1/z$ as $| z | \to \infty$, to obtain an integral expression for $\Phi (z)$ (see \cite{MinwallaWadia:2013,ML:2019}). Explicitly:
		\begin{equation*}
			\Phi (z) = \oint_{\mathcal{C}} \frac{ \dd u }{ 2 \pi  } \frac{ W(u) }{ (u-z) h(u) } ,
		\end{equation*}
		where we denoted for shortness $W(u)$ the left hand side of eq. \eqref{appeq:SPEBaikH}, and the contour $\mathcal{C}$ encloses the branch cut $\mathcal{L}$, the arc along the unit circle. See Figure \ref{fig:contourU}.
		\begin{figure}[htb]
			\centering
			\includegraphics[width=0.4\textwidth]{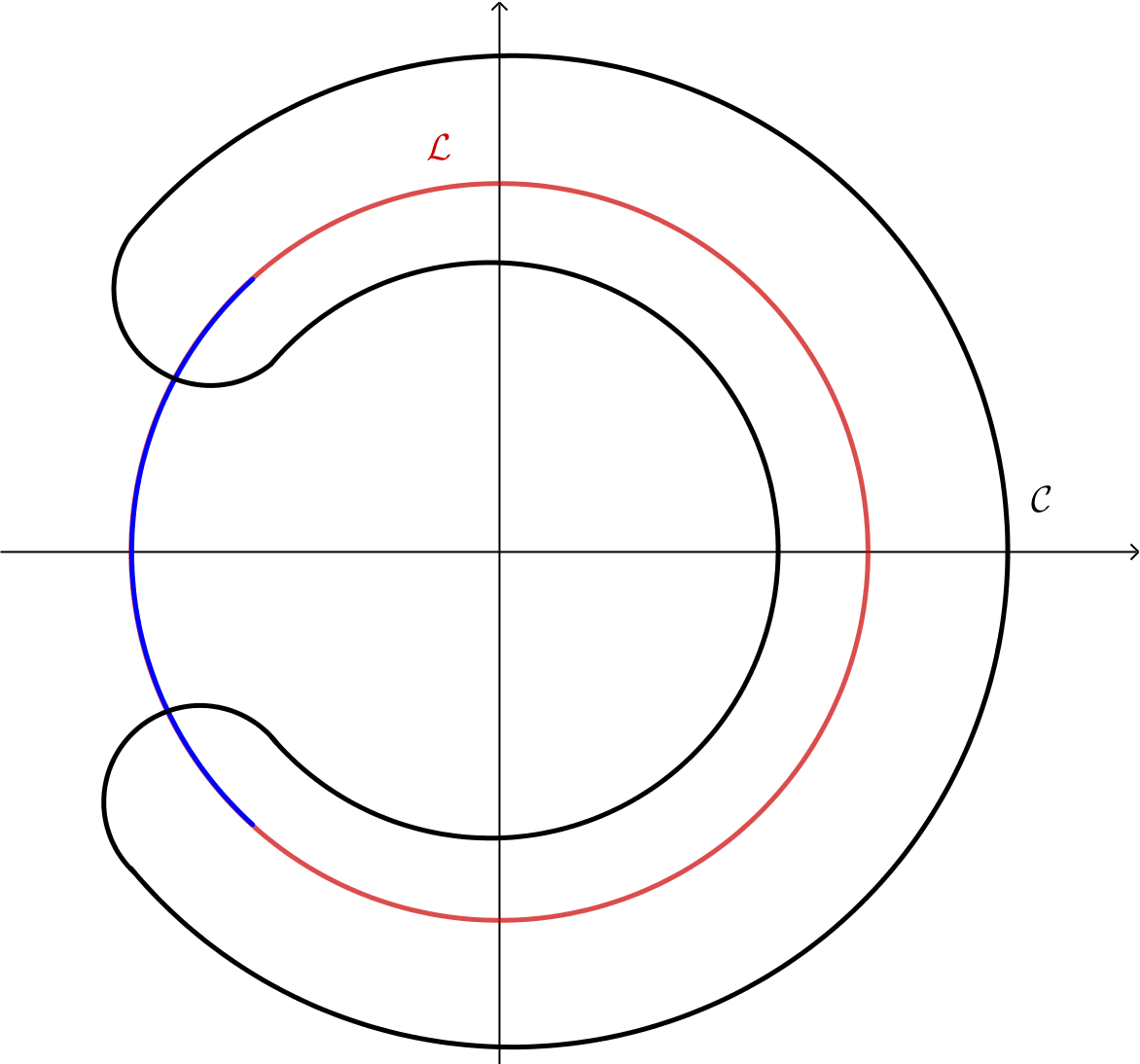}
				\caption{Integration contour in the complex plane. The red line is the cut $\mathcal{L}$, the black curve represents a choice of the contour $\mathcal{C}$. The blue arc on the unit circle, complementary to $\mathcal{L}$, is the gap, where the resolvent $\Psi (z)$ is continuous.}
		\label{fig:contourU}
		\end{figure}\par
		At this stage, we notice that the left hand side of eq. \eqref{appeq:SPEBaikH} has poles but no branch cuts, so we can manipulate the integration contour $\mathcal{C}$ in a convenient way. We arrive to:
		\begin{equation*}
				\Phi (z) = \mathcal{I}_1 (z) + \mathcal{I}_2 (z) + \mathcal{I}_3 (z) ,
		\end{equation*}
		where the three contributions are:
		\begin{align*}
				\mathcal{I}_1 (z) & = -\frac{ \gamma t }{h (z) } \left( \frac{z}{1-tz} - \frac{z^{-1}}{1-tz^{-1}} \right) , \\
				\mathcal{I}_2 (z) & = \sum_{  z_p  } \mathrm{Res}_{u=z_p} \frac{ - \gamma t }{h (u) (u-z) } \left( \frac{u}{1-tu} - \frac{u^{-1}}{1-tu^{-1}} \right) , \\
				\mathcal{I}_3 (z) & = \lim_{R \to \infty} \oint_{\mathcal{C}_R} \frac{ \dd u}{ 2 \pi \ii }  \frac{ \gamma t }{h (u) (u-z) } \left( \frac{u}{1-tu} - \frac{u^{-1}}{1-tu^{-1}} \right) ,
		\end{align*}
		The sum in $\mathcal{I}_2$ runs over the poles $\left\{ z_p \right\}$ of the derivative of the potential (that is, the poles of the left hand side of eq. \eqref{appeq:SPEBaikH}), in this case $z_p = t^{\pm 1}$, and $\mathcal{C}_R$ is a large circle of radius $R$. We have $\mathcal{I}_3 (z) =0$ and $\mathcal{I}_1 (z) h(z)$ yields no discontinuity, so it is irrelevant for the evaluation of $\psi (z)$. The unique relevant contribution comes from:
		\begin{equation}
			\begin{aligned}
				\mathcal{I}_2 (z) &= - \gamma t \left\{  \mathrm{Res}_{u=t} \frac{ \left( \frac{u}{1-tu} - \frac{u^{-1}}{1-tu^{-1}} \right)  }{ h(u) (u-z) } + \mathrm{Res}_{u=1/t} \frac{ \left( \frac{u}{1-tu} - \frac{u^{-1}}{1-tu^{-1}} \right)  }{ h(u) (u-z) } \right\} \\
					& = \gamma \left[ \frac{ t  }{ h(t) (t-z)} + \frac{ t^{-1} }{h (t^{-1} ) (t^{-1} - z)} \right] = \frac{ - \gamma t }{h (t) } \left( \frac{1}{1-tz} + \frac{ z^{-1} }{1-t z^{-1} } \right) ,
			\end{aligned}
		\end{equation}
		where for the last equality we used $h(t^{-1}) = - t^{-1} h(t)$. In general, on the real axis $h(x)>0$ if $x>1$ and $h (x)<0 $ if $x<1$. Therefore, since $0<t<1$, we bare in mind that $-h(t) >0$.\par
		Plugging the expression for $\Phi (z)$ into $\Psi (z)$ and taking its discontinuity along the arc $\mathcal{L}$, we arrive to:
		\begin{equation}
		\label{appeq:psiintermed}
			\psi_H ( e^{\ii \varphi } ) =  - \frac{ 2 \gamma t (1-t) }{\pi h(t)} \left( \frac{ \cos \frac{\varphi}{2} }{ (1-t)^2 +4t \left( \sin \frac{\varphi}{2} \right)^2 } \right)\sqrt{ \left( \sin \frac{\phi_0}{2} \right)^2 -  \left( \sin \frac{\varphi}{2} \right)^2  } .
		\end{equation}
		The boundary $\phi_0$ of the support is fixed by normalization: 
		\begin{equation*}
			1 = \int_{- \phi_0} ^{\phi_0} \dd \rho_{H}  (\varphi) = \gamma \left( \frac{ 1-t}{h(t)} +1 \right) ,
		\end{equation*}
		which provides
		\begin{equation*}
				h(t) = - \frac{ \gamma (1-t)}{(\gamma-1)}
		\end{equation*}
		and hence, writing $h(t)$ explicitly, we obtain:
		\begin{equation*}
				\left( \sin \frac{\phi_0}{2} \right)^2 = \frac{ (1-t)^2 (2 \gamma-1) }{ 4t (\gamma-1)^2 } .
		\end{equation*}
		Also, plugging the expression for $h (t)$ in \eqref{appeq:psiintermed}, we arrive to the final expression
		\begin{equation*}
				\rho_H (\varphi) = \frac{2 (\gamma-1) t }{\pi} \left( \frac{ \cos \frac{\varphi}{2} }{ (1-t)^2 +4t \left( \sin \frac{\varphi}{2} \right)^2 } \right)\sqrt{ \left( \sin \frac{\phi_0}{2} \right)^2 -  \left( \sin \frac{\varphi}{2} \right)^2  } .
		\end{equation*}

		\subsection{Solution for $\ZUsymE$ at weak coupling}
		\label{app:ZUEweak}
			In this Appendix we solve the saddle point equation \eqref{SPEBaikE}, which we rewrite here for completeness:
			\begin{equation}
				\label{appeq:SPEBaikE}
					- \ii \gamma t \left[ \frac{ e^{\ii \varphi} }{1+t e^{\ii \varphi} } - \frac{ e^{- \ii \varphi} }{1+t e^{- \ii \varphi} } \right] = \mathrm{P} \int \dd \vartheta \rho_E (\vartheta) \cot \left( \frac{ \varphi - \vartheta }{2} \right) .
			\end{equation}
			The analysis will follow closely Appendix \ref{app:ZUHweak} for small coupling and Appendix \ref{app:ZUHstrong} for strong coupling, and we omit most of the details.\par
			We again start with the ansatz $\mathrm{supp} \rho =[- \pi, \pi]$, and manipulate the right hand side of \eqref{appeq:SPEBaikE} exactly as we did in Appendix \ref{app:ZUHweak}, and for the left hand side we use:
			\begin{equation*}
				- \ii \gamma t \left[ \frac{ e^{\ii \varphi} }{1+t e^{\ii \varphi} } - \frac{ e^{- \ii \varphi} }{1+t e^{- \ii \varphi} } \right] =  - 2  \gamma t \sin \varphi \sum_{n=0} ^{\infty} U_n ( \cos \varphi ) (-t)^n = 2 \gamma \sum_{n=1} ^{\infty} \sin (n \varphi ) (-t)^n .
			\end{equation*}
			Then, eq. \eqref{appeq:SPEBaikE} becomes
			\begin{equation*}
				 2 \gamma \sum_{n=1} ^{\infty} \sin (n \varphi ) (-t)^n  = - 2 \pi \sum_{n=1} ^{\infty} \left[ a_n \sin (n \varphi) - b_n \cos (n \varphi) \right] ,
			\end{equation*}
			with $a_n, b_n$ the Fourier coefficients of $\rho_{E} (\varphi)$, and we immediately obtain:
			\begin{equation*}
				\rho_{E} (\varphi) = \frac{1}{2 \pi} \left[ 1 - 2 \gamma \sum_{n=1} ^{\infty} \cos (n \varphi) (-t)^n \right] ,
			\end{equation*}
			where, as usual, $a_0 = 1/( 2 \pi)$ is fixed by normalization. We recognize the generating function of Chebyshev polynomials of the first kind, using the standard identity $T_n (\cos (\varphi))= \cos (n \varphi)$, and arrive to the final expression for the eigenvalue density
			\begin{equation*}
				\rho_{E} (\varphi) = \frac{1}{2 \pi} \left[ 1 + 2 \gamma t \left( \frac{ \cos \varphi + t }{1 + t^2 +2 t \cos \varphi } \right) \right] .
			\end{equation*}
			The minima are placed at $\varphi = \pm \pi$, and the non-negativity condition $\rho_{E} \ge 0$ holds as long as:
			\begin{equation*}
				 \gamma \le \frac{ 1-t}{2t} =: \gamma_{c,E} .
			\end{equation*}
			We see that, as expected by na\"{i}ve comparison of the integral representation of the matrix models \eqref{eq:BaikH} and \eqref{eq:BaikE}, the eigenvalue densities $\rho_{H}$ and $\rho_{E}$, as well as the respective critical points, are related through $(\gamma, t) \leftrightarrow (- \gamma, -t) $.

		\subsection{Solution for $\ZUsymE$ at strong coupling}
		\label{app:ZUEstrong}
			The procedure at strong coupling is as in Appendix \ref{app:ZUHstrong}, and we avoid the technical details here. We assume a one-cut solution supported on $[- \phi_0, \phi_0]$, $0<\phi_0 < \pi$, and introduce the resolvent
			\begin{equation*}
				\Psi (z) = \int_{\mathcal{L}} \frac{ \dd u }{ \ii u} \frac{z+u}{z-u} \psi_E (u) , \qquad z \in \C \setminus \mathcal{L},
			\end{equation*}
			with integration contour $\mathcal{L}$ meant to be the arc along the unit circle connecting $e^{- \ii \phi_0} $ to $e^{\ii \phi_0}$, and the function $\psi_E (u)$ being the continuation of $\rho_E$ in $\C$, with $\psi (e^{\ii \varphi}) = \rho_{E} (\varphi)$. We introduce, as in Appendix \ref{app:ZUHstrong}, the complex functions $h(z)$ and $\Phi (z)$, and the saddle point equation \eqref{appeq:SPEBaikE} becomes a discontinuity equation for $\Phi (z)$. The only relevant contribution (for the computation of $\rho_E$) to $\Phi (z)$ is:
			\begin{equation*}
			\begin{aligned}
			\mathcal{I}_2 (z) &= - \gamma t \left\{  \mathrm{Res}_{u=-t} \frac{ \left( \frac{u}{1+tu} - \frac{u^{-1}}{1+tu^{-1}} \right)  }{ h(u) (u-z) } + \mathrm{Res}_{u=- 1/t} \frac{ \left( \frac{u}{1+tu} - \frac{u^{-1}}{1+tu^{-1}} \right)  }{ h(u) (u-z) } \right\} \\
				& = - \gamma \left[ \frac{ t  }{ h(-t) (t+z)} + \frac{ t^{-1} }{h (-t^{-1} ) (t^{-1} + z)} \right]  = \frac{ - \gamma t }{h (-t) } \left( \frac{1}{1+tz} + \frac{ z^{-1} }{1+t z^{-1} } \right) ,
			\end{aligned}
			\end{equation*}
			where for the last equality we used $h(-t^{-1}) = t^{-1} h(-t)$, with both understood to be negative. We get:
			\begin{equation}
			\label{eq:psiBaikEstrong}
				\rho_{E} (\varphi) = - \frac{ 2 \gamma t (1+t) }{\pi h(-t)} \left( \frac{ \cos \frac{\varphi}{2} }{ (1+t)^2 -4t \left( \sin \frac{\varphi}{2} \right)^2 } \right)\sqrt{ \left( \sin \frac{\phi_0}{2} \right)^2 -  \left( \sin \frac{\varphi}{2} \right)^2  } ,
			\end{equation}
			and $\sin (\phi_0 /2)$ is fixed by the normalization:
			\begin{equation*}
				1 = \int_{- \phi_0} ^{\phi_0} \rho_{E} (\varphi) \dd \varphi = - \gamma \left( \frac{ 1+t}{h(-t)} +1 \right) .
			\end{equation*}
			Thus,
			\begin{equation*}
				h(-t) = - \frac{ \gamma (1+t)}{\gamma+1} \quad \Longrightarrow \quad \left( \sin \frac{\phi_0}{2} \right)^2 = \frac{ (1+t)^2 (2 \gamma+1) }{ 4t (\gamma+1)^2 } .
			\end{equation*}
			Combining this latter expression with eq. \eqref{eq:psiBaikEstrong}, we finally arrive to:
			\begin{equation*}
				\rho_E (\varphi) = \frac{ 2 }{\pi}  (\gamma+1)t \left[ \frac{ \cos \left( \frac{\varphi}{2} \right) }{ (1+t)^2 -4t  \left( \sin \left( \frac{\varphi}{2} \right) \right)^2 } \right] \sqrt{ \left( \sin \frac{\phi_0}{2} \right)^2 -  \left( \sin \frac{\varphi}{2} \right)^2  } ,
			\end{equation*}
			supported in $[- \phi_0, \phi_0]$.

	\subsection{Solution for $\ZguH$ at weak coupling}
	\label{app:ZgenHweak}
			This Appendix is dedicated to the solution, in the weak coupling regime, to the saddle point equation \eqref{eq:SPEgenBaikH}, which we rewrite here for convenience:
			\begin{equation}
			\label{appeq:speZgH}
				- \ii \gamma t \left[ \frac{ (1-v) e^{\ii \varphi} }{1-t e^{\ii \varphi} } - \frac{ (1+v) e^{- \ii \varphi} }{1-t e^{- \ii \varphi} } \right] = \mathrm{P} \int \dd \vartheta \rho_{H} (\vartheta) \cot \left( \frac{ \varphi - \vartheta }{2} \right) .
			\end{equation}
			The procedure is similar to the one adopted in Appendix \ref{app:ZUHweak}, but now the parameter $v$ makes the left hand side complex-valued. Therefore, we expand the eigenvalue density in the exponential form:
			\begin{equation*}
				\rho_{H} (\vartheta) = \sum_{n \in \Z} \rho_{H,n} e^{\ii n \vartheta } ,
			\end{equation*}
			while we also use\footnote{The expression holds for $\varphi \ne \vartheta$, which is guaranteed by the principal value of the integral.}:
			\begin{equation*}
				\cot \left( \frac{ \varphi - \vartheta }{2} \right) = - \ii \sum_{n=1} ^{\infty} \left[ e^{\ii n (\varphi - \vartheta )} - e^{- \ii n (\varphi - \vartheta )}  \right] .
			\end{equation*}
			The right hand side of \eqref{appeq:speZgH} then becomes:
			\begin{equation*}
				\mathrm{P} \int \dd \vartheta \rho_H (\vartheta) \cot \left( \frac{ \varphi - \vartheta }{2} \right)  = - 2 \pi \ii \sum_{n=1} ^{\infty} \left[ \rho_{H,n} e^{\ii n \varphi } - \rho_{H,-n} e^{- \ii n \varphi } \right] .
			\end{equation*}
			For the left hand side we can expand the geometric series and obtain:
			\begin{equation*}
				- \ii \gamma t \left[ \frac{ (1-v) e^{\ii \varphi} }{1-t e^{\ii \varphi} } - \frac{ (1+v) e^{- \ii \varphi} }{1-t e^{- \ii \varphi} } \right]  = - \ii \gamma \sum_{n=1} ^{\infty} t^n \left[ (1-v) e^{\ii n \varphi } - (1+v) e^{- \ii n \varphi } \right] .
			\end{equation*}
			Putting all together, and taking into account the normalization condition that fixes $\rho_0$, we arrive to:
			\begin{equation*}
				\rho_H (\varphi) =  \frac{1}{2\pi} \left\{ 1 + 2 \gamma \sum_{n=1} ^{\infty} t^n \left[ \cos (n \varphi) - \ii v \sin (n \varphi) \right] \right\} .
			\end{equation*}
			We see that the parameter $v$, which controls the asymmetry, introduces an imaginary part in the eigenvalue density. Since $-1<v<1$, the relevant minima are those of the real part, thus the critical value is the same as in the symmetric case:
			\begin{equation*}
				\gamma_{c ,H} = \frac{1+t}{2t} ,
			\end{equation*}
			and above this value the solution ceases to be valid.

	\subsection{Solution for $\ZguH$ at strong coupling}
	\label{app:ZgenHstrong}
		
			When $\gamma > \gamma_{c,H}$ we have to look for a different solution to the saddle point equation \eqref{appeq:speZgH}. At strong coupling, the procedure is exactly the same as in Appendix \ref{app:ZUHstrong}, since we were already working in the complex plane, so the complexification of the left hand side of eq. \eqref{appeq:speZgH} for $v \ne 0$ does not alter the strategy. We assume a gapped one-cut solution supported on $[- \phi_0, \phi_0]$ and introduce the resolvent
			\begin{equation*}
				\Psi (z) = \int_{\mathcal{L}} \frac{ \dd u }{ \ii u} \frac{z+u}{z-u} \psi (u) , \qquad z \in \C  \setminus \mathcal{L} , \ \mathcal{L} \subset \mathbb{S}^1 .
			\end{equation*}
			Everything goes trough exactly as in Appendix \ref{app:ZUHstrong}, except for the factors of $(1 \pm v)$ appearing in the left hand side of eq. \eqref{appeq:speZgH}. One easily finds:
			\begin{equation*}
				\psi (e^{\ii \varphi} ) = - \gamma \frac{t}{ \pi h(t) } \left[ \frac{ (1-v) e^{\ii \varphi /2 } }{ 1-t e^{\ii \varphi } } + \frac{ (1+v) e^{- \ii \varphi /2 } }{ 1-t e^{- \ii \varphi } } \right] \sqrt{ \left( \sin \left( \frac{ \phi_0 }{2} \right) \right)^2 - \left( \sin \left( \frac{ \varphi }{2} \right) \right)^2 } .
			\end{equation*}
			The normalization fixes $\sin \left( \phi_0 /2 \right)^2$, through:
			\begin{equation*}
				1 = \int_{- \phi_0} ^{\phi_0} \dd \varphi \rho_H (\varphi ) =  \gamma \left( \frac{1-t}{h(t)} +1 \right) .
			\end{equation*}
			Notice that, due to the parity of the integral, the result is independent of $v$, in particular it is equal to the symmetric case ($v=0$). The final expression for the eigenvalue density:
			\begin{equation*}
				\rho_H (\varphi) = \frac{2 t  ( \gamma -1)  }{\pi (1-t) } \left[   \frac{  (1-t) \cos \left( \frac{ \varphi }{2} \right) - \ii v (1+t) \sin \left( \frac{ \varphi }{2} \right) }{  (1-t)^2 + 4 t \left( \sin \left( \frac{ \varphi }{2} \right) \right)^2 } \right] \sqrt{ \left( \sin \left( \frac{ \phi_0 }{2} \right) \right)^2 - \left( \sin \left( \frac{ \varphi }{2} \right) \right)^2 } ,
			\end{equation*}
			with $\sin (\phi_0 /2)$ the same as in \cite{Baik:2000} and Appendix \ref{app:ZUHstrong}, that is
			\begin{equation*}
				\left( \sin \frac{\phi_0}{2} \right)^2 = \frac{ (1-t)^2 (2 \gamma-1) }{ 4t (\gamma-1)^2 } .
			\end{equation*}

	\subsection{Solution for $\ZguE$ at weak coupling}
	\label{app:ZgenEweak}
			Here we solve the large $K$ limit of the second matrix model, eq. \eqref{eq:ZgenBaikE}, in the general case. The saddle point equation \eqref{eq:SPEgenBaikE} is:
			\begin{equation}
			\label{appeq:speZgE}
				- \ii \gamma t  \left( \frac{ (1-v) e^{\ii \varphi} }{1+t e^{\ii \varphi}} -  \frac{ (1+v) e^{ - \ii \varphi} }{1+t e^{- \ii \varphi}} \right) = \mathrm{P} \int \dd \vartheta \rho _{E} (\vartheta)  \cot \frac{ \varphi - \vartheta }{2}  .
			\end{equation}
			The procedure is as in Appendix \ref{app:ZgenHweak}. The right hand side reads:
			\begin{equation*}
				- 2 \pi \ii \sum_{n=1} ^{\infty} \left[ \rho_{E,n} e^{\ii n \varphi } - \rho_{E,-n} e^{- \ii n \varphi } \right] ,
			\end{equation*}
			while, expanding the geometric series, the left hand side becomes:
			\begin{equation*}
				\ii \gamma \sum_{n=1} ^{\infty} (-t)^n \left[ (1-v) e^{\ii n \varphi} - (1+v) e^{- \ii n \varphi} \right] . 
			\end{equation*}
			This gives:
			\begin{equation*}
				\rho_E (\varphi ) =  \frac{1}{2 \pi} \left[1 + 2 \gamma t \left( \frac{ \cos (\varphi ) +t  - \ii v \sin (\varphi ) }{ 1 + t^2 + 2 t \cos (\varphi)  } \right) \right] .
			\end{equation*}
			Again, the validity of this solution extends up to:
			\begin{equation*}
				\gamma_{c,E} = \frac{1-t}{2t} ,
			\end{equation*}
			as in the symmetric case.

	\subsection{Solution for $\ZguE$ at strong coupling}
	\label{app:ZgenEstrong}
			When $\gamma > \gamma_{c,E}$, we drop the assumption $\mathrm{supp} \rho_E = [-\pi, \pi ]$ and look for a one-cut solution with a gap. The procedure is clear from the previous Appendices, and the result is:
			\begin{equation*}
				\rho_E (\varphi ) = - \frac{ \gamma t }{\pi h(-t) } \left[ \frac{ (1-v) e^{\ii \varphi /2 } }{ 1+t e^{\ii \varphi} } + \frac{ (1+v) e^{- \ii \varphi /2 } }{ 1+t e^{- \ii \varphi} } \right] \sqrt{ \left( \sin \left( \frac{ \phi_0 }{2} \right) \right)^2  - \left( \sin \left( \frac{ \varphi}{2} \right) \right)^2} ,
			\end{equation*}
			which, after imposing normalization and rewriting in terms of trigonometric functions, gives:
			\begin{equation*}
				\rho_E (\varphi) = \frac{ 2 t (\gamma+1) }{\pi (1+t) }  \left[ \frac{ (1+t) \cos \left( \frac{\varphi}{2} \right) - \ii v (1-t) \sin  \left( \frac{\varphi}{2} \right) }{ (1+t)^2 -4t  \left( \sin\left( \frac{\varphi}{2} \right)\right)^2 } \right] \sqrt{ \left( \sin \left( \frac{\phi_0}{2} \right) \right)^2 -  \left( \sin \left( \frac{\varphi}{2} \right) \right)^2  } ,
			\end{equation*}
			supported on
			\begin{equation*}
				\mathrm{supp} \rho_E = [- \phi_0, \phi_0 ] , \quad \left( \sin \left( \frac{\phi_0}{2} \right) \right)^2 = \frac{ (1+t)^2 (2 \gamma +1) }{4t (\gamma +1)} .
			\end{equation*}

	\section{Free energies}
	\label{app:FEcalc}
		This Appendix contains the calculations to obtain the free energy of the unitary matrix models considered in the main text.
		
		\subsection{Free energy for $\ZUsymH$ and $\ZUsymE$ at weak coupling}
		\label{app:Fbothweak}
			The simplest way to obtain the free energy is to evaluate its derivative with respect to the parameter $t$ and then integrate. In the weak coupling phase, $0 \le \gamma \le \gamma_c$, we can use the boundary condition $\ZU(\gamma=0)=1$, which follows immediately from the normalization of the Haar measure on $U(N)$. At strong coupling we use the continuity of $\log \ZU (\gamma)$ at $\gamma= \gamma_c$.\par
			For what concerns $\mfsymH$, at weak coupling we have:
			\begin{equation*}
				\frac{ \dd \mfsymH}{ \dd t } (\gamma \le \gamma_{c,H} ) =  \gamma \int \dd \varphi \rho_H (\varphi) \left[ \frac{ e^{\ii \varphi} }{1-t e^{\ii \varphi}} + \frac{ e^{- \ii \varphi} }{1-t e^{- \ii \varphi}} \right] = \frac{ 2 \gamma^2 t }{1-t^2} .
			\end{equation*}
			Integrating with boundary condition $\ZU^{H} (\gamma=0)=1$ gives
			\begin{equation}
				\mfsymH (\gamma \le \gamma_{c,H} ) = - \gamma^2 \log (1-t^2) .
			\end{equation}\par
			For $\mfsymE$ at weak coupling we have:
			\begin{equation*}
				\frac{ \dd \mfsymE }{ \dd t } (\gamma \le \gamma_{c,E} )  =  \gamma \int \dd \varphi \rho_E (\varphi) \left[ \frac{ e^{\ii \varphi} }{1+t e^{\ii \varphi}} + \frac{ e^{- \ii \varphi} }{1+t e^{- \ii \varphi}} \right]  = \frac{ 2 \gamma^2 t }{1-t^2} ,
			\end{equation*}
			which, after integration, gives
			\begin{equation}
				\mfsymE (\gamma \le \gamma_{c,E} ) = - \gamma^2 \log (1-t^2) .
			\end{equation}
			In particular, we see that the free energies of the two models are equal in the weak coupling phase. This is expected, since, in the weak coupling phase, the free energy equals the result provided by Szeg\H{o} theorem, \emph{i.e.} the limit without scaling (see e.g. \cite{ML:2019} for a proof of this statement), and it is well known that the limits of the models \eqref{eq:BaikH} and \eqref{eq:BaikE} without scaling are equal.

		\subsection{Free energy for $\ZUsymH$ and $\ZUsymE$ at strong coupling}
		\label{app:Fbothstrong}
			We now pass to the strong coupling phase, and use the form of the eigenvalue densities $\rho_H$ and $\rho_E$ at $gamma > \gamma_c$.\par
			Starting with $\mfsymH$, we have:
			\begin{equation*}
			\begin{aligned}
				\frac{ \dd \mfsymH }{ \dd t } (\gamma > \gamma_{c,H} ) & =  \gamma \int \dd \varphi \rho_H (\varphi) \left[ \frac{ e^{\ii \varphi} }{1-t e^{\ii \varphi}} + \frac{ e^{- \ii \varphi} }{1-t e^{- \ii \varphi}} \right] \\
					& = \frac{ 4 \gamma (\gamma-1) t }{\pi} \int_{- \phi_0} ^{\phi_0} \dd \varphi \frac{ \cos \frac{\varphi}{2}  \sqrt{ \left( \sin \frac{ \phi_0}{2} \right)^2 - \left( \sin \frac{ \varphi}{2} \right)^2 } }{ (1-t)^2 +4t \left( \sin \frac{ \varphi}{2} \right)^2 } \left[ \frac{ \cos \varphi - t }{ (1-t)^2 +4t \left( \sin \frac{ \varphi}{2} \right)^2}  \right] \\
					& = \frac{ \gamma (\gamma -1) }{\pi t } \int_{-1} ^1 \dd y \frac{ \sqrt{1-y^2} \left( \frac{ 2t (\gamma-1)^2 }{(1-t)(2 \gamma-1)} -y^2 \right) }{ \left[ \frac{ (\gamma-1)^2 }{(2\gamma-1)} + y^2 \right]^2 } \\
					& = - \frac{ 1 + t - 4 \gamma t }{2t (1-t) } ,
			\end{aligned}
			\end{equation*}
			where we used the change of variables $y= \frac{ \sin (\varphi/2) }{\sin (\phi_0 /2)}$ and used the explicit form of $\sin (\phi_0 /2) ^2$. Immediate integration gives:
			\begin{equation*}
				\mfsymH (\gamma > \gamma_{c,H}) = - (2 \gamma -1) \log (1-t) - \frac{1}{2} \log (t) + \mathcal{C}_{H} (\gamma) ,
			\end{equation*}
			with $\mathcal{C}_{H} (\gamma)$ a $t$-independent integration constant fixed by continuity at $\gamma = \gamma_{c,H}$.\par
			For $\mfsymE$ the calculations are almost the same, except for the change of sign in front of all factors of $t$ and $\gamma$. We get:
			\begin{equation*}
			\begin{aligned}
				\frac{ \dd \mfsymE }{ \dd t } (\gamma > \gamma_{c,E} ) & =  \gamma \int \dd \varphi \rho_E (\varphi) \left[ \frac{ e^{\ii \varphi} }{1+t e^{\ii \varphi}} + \frac{ e^{- \ii \varphi} }{1+t e^{- \ii \varphi}} \right] \\
					& = \frac{ \gamma (\gamma +1) }{\pi t } \int_{-1} ^1 \dd y \frac{ \sqrt{1-y^2} \left( \frac{ 2t (\gamma+1)^2 }{(1+t)(2 \gamma+1)} -y^2 \right) }{ \left[ \frac{ (\gamma+1)^2 }{(2\gamma+1)} - y^2 \right]^2 } \\
					& = - \frac{ 1 - t - 4 \gamma t }{2t (1+t) } .
			\end{aligned}
			\end{equation*}
			Notice that the final line here could not be inferred from the final line of $\mfsymH$ with reversed signs. After integration we obtain the result:
			\begin{equation}
				\mfsymE (\gamma > \gamma_{c,E}) =  (2 \gamma +1) \log (1+t) - \frac{1}{2} \log (t) + \mathcal{C}_{E} (\gamma) ,
			\end{equation}
			where $\mathcal{C}_{E} (\gamma)$ is $t$-independent and fixed by continuity at the critical value $\gamma_{c,E}$.

		\subsection{Free energy for $\ZguH$ and $\ZguE$ }
		\label{app:FEgenboth}
			The free energies for the matrix models \eqref{eq:ZgenBaikH} and \eqref{eq:ZgenBaikE} are computed at weak coupling in exactly the same manner as in Appendix \ref{app:Fbothweak}. It is easy to check that, thanks to the parity of the integral, the free energies receive contribution from the terms proportional to $v$ but they remain real. Direct computations give:
			\begin{equation*}
				\frac{ \dd \mfgenuH }{ \dd t } (\gamma < \gamma_{c,H} ) = \frac{2  \gamma^2 t}{1-t^2} (1-v^2) = \frac{ \dd \mfgenuE }{ \dd t } (\gamma < \gamma_{c,E} ) .
			\end{equation*}\par
			At strong coupling we get:
			\begin{equation*}
			\begin{aligned}
				\frac{ \dd \mfgenuH }{ \dd t } (\gamma > \gamma_{c,H} ) & = \gamma \int_{- \phi_0} ^{\phi_0} \dd \varphi \rho_H (\varphi ) \left[ \frac{ (1-v) e^{\ii \varphi } }{ 1-te^{\ii \varphi } } + \frac{ (1+v) e^{- \ii \varphi } }{ 1-te^{- \ii \varphi } } \right] \\ 
				& = \frac{ 4 \gamma (\gamma-1) t }{\pi (1-t) } \int_{- \phi_0} ^{\phi_0} \dd \varphi \frac{ \left(  (1-t) \cos \frac{\varphi}{2} - \ii v (1+t) \sin \frac{\varphi}{2} \right) }{ \left[ (1-t)^2 +4t \left( \sin \frac{ \varphi}{2} \right)^2 \right]^2 }  \\
				& \times \sqrt{ \left( \sin \frac{ \phi_0}{2} \right)^2 - \left( \sin \frac{ \varphi}{2} \right)^2 }  \left( \cos \varphi - t - \ii v \sin \varphi \right) \\
				& = \frac{ \dd \mfsymH }{ \dd t } (\gamma > \gamma_{c,H} ) - v^2 \frac{ 16 \gamma (\gamma-1) t (1+t) }{\pi (1-t)} \int_{- x_0} ^{x_0} \dd x \frac{ x^2 \sqrt{ x_0 ^2 - x^2} }{ \left[ (1-t)^2 +4t x^2 \right]^2 } \\
				& = \frac{ \dd \mfsymH }{ \dd t } (\gamma > \gamma_{c,H} ) - \frac{v^2 (1+t) }{2t (1-t) } ,
			\end{aligned}
			\end{equation*}
			where in the second line we changed variables $x= \sin (\varphi /2)$, with boundary at $x_0 = \sin ( \phi_0 /2)$, and in the last line we used the explicit form of $x_0^2$ to simplify the resulting expression. Notice that the term proportional to $v^2$ does not depend explicitly on $\gamma$.\par
			Following the same steps, the free energy $\mfgenuE$ is computed at strong coupling as:
			\begin{equation*}
			\begin{aligned}
				\frac{ \dd \mfgenuE }{ \dd t } (\gamma > \gamma_{c,E} ) & =  \gamma \int_{- \phi_0} ^{\phi_0} \dd \varphi \rho_H (\varphi ) \left[ \frac{ (1-v) e^{\ii \varphi } }{ 1-te^{\ii \varphi } } + \frac{ (1+v) e^{- \ii \varphi } }{ 1-te^{- \ii \varphi } } \right] \\ 
				& = \frac{ \dd \mfsymE }{ \dd t } (\gamma > \gamma_{c,E} ) - v^2 \frac{ 16 \gamma (\gamma+1) t (1-t) }{\pi (1+t) } \int_{- x_0} ^{x_0} \dd x \frac{ x^2 \sqrt{ x_0 ^2 - x^2} }{ \left[ (1+t)^2 -4t x^2 \right]^2 } \\
				&=  \frac{ \dd \mfsymE }{ \dd t } (\gamma > \gamma_{c,E} ) - \frac{v^2 (1-t) }{2t (1+t) } .
			\end{aligned}
			\end{equation*}

		\section{Solution to the saddle point equations on the real line}
		Here we consider the solution to the saddle point equation obtained from the large $N$ limit of the matrix model after stereographic projection on the real line, as described in Section \ref{sec:largeKromanov}. We focus on $\mathcal{Z}_{E,\mathrm{stereo.}} ^{\mathrm{sym.}}$ and $\mathcal{Z}_{E,\mathrm{stereo.}} ^{\mathrm{gen.}}$, and the solutions for the $H$-models can be obtained in an analogous way.
		
		\subsection{Solution for  $\mathcal{Z}_{E,\mathrm{stereo.}} ^{\mathrm{sym.}}$}
		\label{app:ZEstereosym}
			We consider first the matrix model $\mathcal{Z}_{E,\mathrm{stereo.}} ^{\mathrm{sym.}}$ defined in eq. \eqref{eq:stereosymE}. In the limit in which the number of variables, $K$, is large, the leading contribution comes from the eigenvalue density $\rho_{\mathrm{stereo.}}$ that solves the saddle point equation \eqref{eq:stereoSPEsym}, which we report here for convenience:
			\begin{equation}
			\label{appeq:stereosymSPE}
				\mathrm{P} \int \dd y \frac{ \rho_{\mathrm{stereo.}} (y) }{x-y} = W(x) , \qquad x \in \R ,
			\end{equation}
			with
			\begin{equation*}
				W(x) = - \frac{ \gamma x }{ x^2 + t_0 ^2 } + \frac{( \gamma + 1) x }{x^2 + 1} , \qquad t_0:= \frac{1+t}{1-t} .
			\end{equation*}
			We proceed following standard methods for the analysis of Hermitian matrix models at large $K$, as reviewed for example in \cite{DiFrancesco:1993}. We introduce the resolvent
			\begin{equation*}
				\Psi (z) = \int_{\mathcal{L}} \dd y \frac{ \rho_{\mathrm{stereo.}} (y) }{ z-y} , \qquad z \in \C \setminus \mathcal{L} ,
			\end{equation*}
			where now, assuming a solution with symmetric support $[-A,A]$, the path $\mathcal{L}$ is the segment $[-A,A]$ on the real line. The eigenvalue density is recovered as the jump of the resolvent along $\mathcal{L}$:
			\begin{equation*}
				\Psi_{+} (z) - \Psi_{-} (z) = - 2 \pi \ii \rho_{\mathrm{stereo.}} (x), \qquad x \in [-A, A] ,
			\end{equation*}
			with $\pm$ meaning the limit taken approaching the real line from the upper (resp. lower) half plane.\par
			A solution for the resolvent is the following: write
			\begin{equation*}
				\Psi (z) = h(z) \Phi (z), \quad h(z) = \sqrt{ (-A-z) (A-z)} ,
			\end{equation*}
			and the saddle point equation \eqref{appeq:stereosymSPE} becomes a discontinuity equation for $\Phi(z)$, with solution:
			\begin{equation*}
				\Phi (z) = \oint _{\mathcal{C}} \frac{ \dd u }{ 2 \pi \ii } \frac{ W(u) }{ (z-u) h(u) } ,
			\end{equation*}
			where the contour $\mathcal{C}$ is a closed curve surrounding the cut $\mathcal{L}$, as in Figure \ref{fig:contourH}.
			\begin{figure}[hb]
				\centering
				\includegraphics[width=0.5\textwidth]{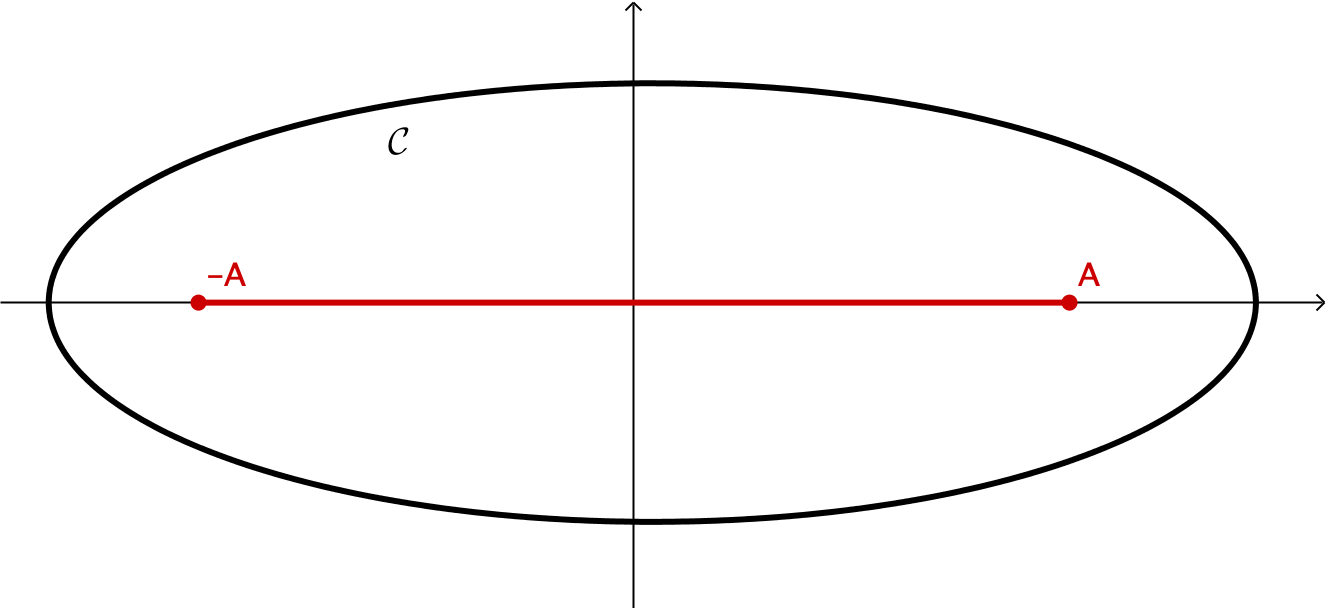}
				\caption{Integration contour in the complex plane. The red line is the cut $\mathcal{L}$, the black curve represents a choice of contour $\mathcal{C}$.}
			\label{fig:contourH}
			\end{figure}\par
			$W(u)$ is a meromorphic function, with simple poles at:
			\begin{equation*}
				u = z_p , \quad z_p \in \left\{ \pm \ii , \ \pm \ii t_0 \right\} .
			\end{equation*}
			We can thus deform the contour, to avoid the branch cut of the square root, and pick the poles of the integrand. This leaves a residual integral along an infinitely large circle. This latter contribution vanishes, since $W(u) \sim 1/u$ at large $|u|$. The pole at $u=z$ generates the regular part of $\Psi (z)$, thus yields no contribution to the eigenvalue density. Therefore, the relevant contributions to the $\rho$ arise from the poles of $W(u)$ in the complex plane. For the calculations, one has to be careful with the signs in front of the square roots, according to the definition of $h(z)$. With our conventions,
			\begin{equation*}
			\begin{aligned}
				h( \ii ) & = - \ii \sqrt{ A^2 + 1 } = - h (- \ii ) , \\
				h( \ii t_0 ) & = - \ii \sqrt{ A^2 + t_0 ^2  } = - h (- \ii t_0  ) .
			\end{aligned}
			\end{equation*}
			After simple computations, one arrives to:
			\begin{equation*}
				\sum_{z_p} \mathrm{Res}_{u=z_p} \frac{ W(u) }{ (z-u) \sqrt{ u^2 - A^2 } } = \frac{ \gamma t_0 }{ (z^2 + t_0 ^2 ) \sqrt{ A^2 + t_0 ^2 } } - \frac{ (\gamma+1) }{ (z^2 + 1 ) \sqrt{ A^2 + 1 } } .
			\end{equation*}
			The eigenvalue density is obtained this function multiplied by the jump of $h(z)$:
			\begin{equation*}
				\rho_{\mathrm{stereo.}} (x) = \frac{ \sqrt{A^2 - x^2 }}{\pi} \left[ \frac{ (1+ \gamma ) }{ \sqrt{ A^2 + 1 } (x^2 + 1) } - \frac{ \gamma t_0 }{ \sqrt{ A^2 + t_0 ^2 } (x^2 + t_0 ^2 ) }  \right] .
			\end{equation*}
			The value of the boundary $A$ can be fixed by normalization, or equivalently looking at the large $z$ behaviour of the resolvent $\Psi (z)$. From the definition, one has $\Psi (z) \sim 1/z $, as $z \to \infty$, while from the explicit evaluation, we have:
			\begin{equation*}
			\begin{aligned}
				\Psi (z) & = W(z) + \sqrt{ z^2 - A^2 } \left( \frac{ \gamma t_0 }{ (z^2 + t_0 ^2 ) \sqrt{ A^2 + t_0 ^2 } } - \frac{ (\gamma+1) }{ (z^2 + 1 ) \sqrt{ A^2 + 1 } }  \right) \\
					& = \frac{1}{z} \left( 1 + \frac{\gamma t_0 }{ \sqrt{ A^2 + t_0 ^2 } }   - \frac{ \gamma+1 }{\sqrt{A^2 +1}} \right) + \mathcal{O} (z^{-2}) .
			\end{aligned}
			\end{equation*}
			This provides the solution (recall that $t_0 = (1+t)/(1-t)$)
			\begin{equation*}
				A^2 = \frac{ (2 \gamma + 1 )(1+t)^2 }{ (2 \gamma + 1  - t )( 2 \gamma - 1 + t) } .
			\end{equation*}
			This is a positive quantity, and thus $\rho_{\mathrm{stereo.}}$ is supported on the real line, as long as
			\begin{equation*}
				\gamma > \frac{1-t}{2t} .
			\end{equation*}
			This is consistent with our general analysis: due to the change of topology, the stereographic projection should only provide the gapped phase of the unitary matrix model. We also notice that, undoing the stereographic projection, we obtain:
			\begin{equation*}
				e^{ \ii \phi_0 } = \frac{ 1+ \ii A }{ 1 - \ii A }  \quad \Longrightarrow \quad \left( \sin \frac{ \phi_0 }{2} \right)^2 = \frac{ (2 \gamma +1 ) (1+t)^2 }{ 4 t (\gamma +1 )^2 } ,
			\end{equation*}
			thus the boundary we obtain on the real line is in fact the stereographic projection of the boundary of the model on the unit circle, as expected.

		\subsection{Solution for  $\mathcal{Z}_{E,\mathrm{stereo.}} ^{\mathrm{gen.}}$}
		\label{app:ZEstereogen}
			In this Appendix, we study the large $K$ limit of the matrix model $\mathcal{Z}_{E,\mathrm{stereo.}} ^{\mathrm{gen.}}$ defined in eq. \eqref{eq:stereogenE}. The saddle point equation is:
			\begin{equation}
			\label{appeq:stereogenSPE}
				\mathrm{P} \int \dd y \frac{ \rho_{\mathrm{stereo.}} (y) }{x-y} = W(x) , \qquad x \in \R ,
			\end{equation}
			with $W(x)$ in the present, more general case given by:
			\begin{equation*}
				W(x) = - \frac{ \gamma x }{ x^2 + t_0 ^2 } + \frac{( \gamma + 1) x }{x^2 + 1} + 4 \ii v \gamma \frac{t}{1-t} \frac{ x^2 - t_0 }{ (x^2 +1)(x^2 + t_0 ^2)} .
			\end{equation*}
			From this expression, it is clear that the asymmetry in the integrand introduces an imaginary part in the eigenvalue density, but does not introduce new poles of $W(u)$ in the complex plane.\par
			We proceed as in the previous Appendix, and find:
			\begin{equation*}
				\Psi (z)  = W(z) + h(z) \sum_{z_p} \mathrm{Res}_{u=z_p} \frac{ W(u) }{ (z-u) \sqrt{ u^2 - A^2 } } ,
			\end{equation*}
			where now the residues at $u= \pm \ii, \pm \ii t_0$ yield an extra term, proportional to the parameter $v$. After some simple calculations, we see that this extra contribution is:
			\begin{equation*}
				2 \ii v \gamma \left( \frac{ z }{ (z^2 + t_0 ^2 ) \sqrt{ A^2 + t_0 ^2 } } -  \frac{ z }{ (z^2 + 1 ) \sqrt{ A^2 + 1 } } \right) .
			\end{equation*}
			The final expression for the eigenvalue density is:
			\begin{equation*}
				\rho_{\mathrm{stereo.}} (x) = \frac{ \sqrt{A^2 - x^2 }}{\pi} \left[ \frac{ (1+ \gamma ) + \ii v \gamma x  }{ \sqrt{ A^2 + 1 } (x^2 + 1) } - \frac{ \gamma t_0 + \ii v \gamma x }{ \sqrt{ A^2 + t_0 ^2 } (x^2 + t_0 ^2 ) }  \right] ,
			\end{equation*}
			Imposing normalization, from the parity of the integral we obtain the same value of $A$ as above.
			
		\subsection{Solution for  $\mathcal{Z}_{E,\mathrm{stereo.}} ^{\mathrm{sym.}}$ with modified symbol}
		\label{app:stereo+mono}
			
			We now solve the saddle point equation \eqref{appeq:stereosymSPE} with, on the right hand side, the modified function
			\begin{equation*}
				W(x) = - \frac{ \gamma x }{ x^2 + t_0 ^2 } + \frac{( \gamma + 1) x - 2 b }{x^2 + 1} .
			\end{equation*}
			The procedure follows closely that of Appendix \ref{app:ZEstereosym}, but now we do not assume symmetry of the support for $\rho_{\mathrm{stereo.}}$, and let
			\begin{equation*}
				\mathrm{supp} \rho_{\mathrm{stereo.}} = [-A, B] .
			\end{equation*}
			We therefore modify the definition of the auxiliary function
			\begin{equation*}
				h (z) = \sqrt{ (-A -z) (B-z) } .
			\end{equation*}
			For later convenience, we also introduce the functions 
			\begin{equation}
			\label{appeq:redh}
			\begin{aligned}
				\tilde{h} (x) & = \sqrt{ \lvert x^2 + AB + \ii x (B-A) \rvert } = \left( x^4 +A^2 B^2 + x^2 (A^2 + B^2)  \right)^{\frac{1}{4}} , \\
				\tilde{\theta} (x) &= \frac{1}{2} \arg \left(  x^2 + AB + \ii x (B-A) \right) =  \frac{1}{2} \arctan \frac{ x (B-A ) }{ x^2 + AB } ,
			\end{aligned}
			\end{equation}
			for $x \in \R$. Note also that $\tilde{h} (-x) = \tilde{h} (x)$ and $\tilde{\theta} (- x)= - \tilde{\theta} (x)$. Following Appendix \ref{app:ZEstereosym}, we introduce the resolvent $\Psi (z)$ and arrive to:
			\begin{equation*}
				\Psi (z)  = W(z) + h(z) \sum_{z_p} \mathrm{Res}_{u=z_p} \frac{ W(u) }{ (z-u) h(u) } ,
			\end{equation*}
			with poles at $z_p = \pm \ii, \pm \ii t_0$. Using
			\begin{equation*}
				h( \ii )  = - \ii \tilde{h} (1) e^{ \ii \tilde{\theta} (1)} , \quad   h (- \ii ) = \ii \tilde{h} (1)  e^{ - \ii \tilde{\theta} (1)} , 
			\end{equation*}
			and similarly for $\pm \ii t_0$, we arrive to:
			\begin{equation*}
			\begin{aligned}
				 \sum_{ z_p} \mathrm{Res}_{u=z_p} \frac{ W(u) }{ (z-u) h(u) } & = \frac{ \gamma \left[ t_0 \cos  \tilde{\theta} (t_0) - z \sin  \tilde{\theta} (t_0) \right] }{ \tilde{h} (t_0) (z^2 + t_0 ^2 ) } \\
					&  - \frac{ \left[ (\gamma +1) \cos \tilde{\theta} (1) + 2 b \sin \tilde{\theta} (1)  \right] + z \left[ 2b \cos \tilde{\theta} (1) - (\gamma +1) \sin \tilde{\theta} (1)  \right]  }{ \tilde{h} (1) (z^2 +1) } ,
			\end{aligned}
			\end{equation*}
			and the eigenvalue density is:
			\begin{equation*}
			\begin{aligned}
				\rho_{\mathrm{stereo.}} (x) = \frac{ \sqrt{(x-A)(B-x) }}{\pi} & \left[ \frac{ \gamma \left[ t_0 \cos \tilde{\theta} (t_0) - z \sin  \tilde{\theta} (t_0) \right] }{ \tilde{h} (t_0) (z^2 + t_0 ^2 ) } \right. \\
				& \left. - \frac{ \left[ (\gamma +1) \cos \tilde{\theta} (1) + 2 b \sin \tilde{\theta} (1)  \right] + z \left[ 2b \cos \tilde{\theta} (1) - (\gamma +1) \sin \tilde{\theta} (1)  \right]  }{ \tilde{h} (1) (z^2 +1) }  \right]
			\end{aligned}
			\end{equation*}
			for $x \in [-A,B]$. The boundaries of the support are fixed by imposing the condition $\Psi (z) \sim 1/z$ at large $z$. Using
			\begin{equation*}
				\cos 2 \tilde{\theta} (x) = \frac{ x^2 + AB }{ \tilde{h} (x) ^2 } ,
			\end{equation*}
			order $z^0$ in the expansion of $\Psi (z)$ gives the constraint
			\begin{equation*}
				\frac{\gamma }{2 \tilde{h} (t_0)} \left[ \frac{ t_0 ^2 + AB }{ \tilde{h} (t_0)^2 } -1 \right] -  \frac{1}{2 \tilde{h} (1)} \left[ \left( \frac{ 1 + AB }{ \tilde{h} (1)^2 } \right) \left( 2b - \gamma - 1 \right) +2b + \gamma +1  \right]  = 0 ,
			\end{equation*}
			while from order $z^1$, and using the previous expression to simplify the equation, we get the constraint
			\begin{equation*}
				\frac{\gamma t_0 }{2 \tilde{h} (t_0)} \left[ \frac{ t_0 ^2 + AB }{ \tilde{h} (t_0)^2 } +1 \right] -  \frac{1}{2 \tilde{h} (1)} \left[ \left( \frac{ 1 + AB }{ \tilde{h} (1)^2 } \right) \left( 2b + \gamma + 1 \right) - 2b + \gamma +1  \right]  = 0 .
			\end{equation*}

\end{appendix}

	\bibliographystyle{ourstyle}
	\bibliography{MeixBaik-biblio_revisado}

\end{document}